\documentclass[conference]{IEEEtran}

\usepackage{amssymb,amsthm}
\usepackage{amsmath,amssymb}
\usepackage{bm,amsfonts}
\usepackage{nidanfloat}
\usepackage{multirow}

\DeclareMathOperator*{\bliminf}{\underline{\lim}}
\DeclareMathOperator*{\blimsup}{\overline{\lim}}

\newcommand{\so}{\mathrm{o}}
\newcommand{\lo}{\mathrm{O}}
\newcommand{\ep}{\epsilon}
\newcommand{\de}{\delta}

\newcommand{\si}{\sigma}
\newcommand{\id}{ \mbox{\rm 1}\hspace{-0.5em}\mbox{\rm \small l\,}}

\newcommand{\idx}[1]{\id\left[#1\right]}

\newcommand{\E}{\mathrm{E}}

\newcommand{\hen}[2]{\frac{\partial #1}{\partial #2}}
\newcommand{\n}{\nonumber}
\newcommand{\nn}{\nonumber\\}

\newcommand{\rd}{\mathrm{d}}

\newcommand{\uR}{\underline{\mathbb{R}}}

\newtheorem{lemma}{Lemma}
\newtheorem{theorem}{Theorem}
\newtheorem{proposition}{Proposition}
\newtheorem{definition}{Definition}
\newtheorem{remark}{Remark}

\newcommand{\la}{\lambda}
\newcommand{\La}{\mathrm{\Lambda}}
\newcommand{\e}{\mathrm{e}}

\newcommand{\De}{\Delta}

\newcommand{\bibun}[2]{\frac{\rd #1}{\rd #2}}
\newcommand{\com}{\,,}
\newcommand{\per}{\,.}

\newcommand{\calX}{\mathcal{X}}
\newcommand{\calY}{\mathcal{Y}}

\newcommand{\since}[1]{\quad\left(\mbox{#1}\right)}

\newcommand{\im}{\mathrm{i}}

\newcommand{\bmX}{\bm{X}}
\newcommand{\bmY}{\bm{Y}}
\newcommand{\bmx}{\bm{x}}
\newcommand{\bmy}{\bm{y}}
\newcommand{\bmV}{\bm{V}}

\newcommand{\al}{\alpha}

\newcommand{\inner}[1]{\langle #1\rangle}
\newcommand{\bibunZ}[1]{Z^{(#1)}}
\newcommand{\ua}{\underline{s}}
\newcommand{\oa}{\overline{s}}
\newcommand{\tg}[1]{g_h^{(#1)}}
\newcommand{\tG}[1]{G_h^{(#1)}}
\newcommand{\ug}{\tilde{g}}
\newcommand{\der}[1]{\mu_{#1}}

\newcommand{\pp}{p_+}
\newcommand{\pz}{p_0}
\newcommand{\cs}[1]{c^{(#1)}}

\newcommand{\gap}{\delta}
\newcommand{\gapla}{\gamma_1}
\newcommand{\gapxia}{b_0}
\newcommand{\gapxi}{b_1}
\newcommand{\dephi}{\gamma_2}
\newcommand{\cumu}{{Z_{\mathrm{a}}}}
\newcommand{\cumub}{{\bar{Z}_{\mathrm{a}}}}
\newcommand{\cumui}[1]{{Z_{\mathrm{a},i}}}
\newcommand{\ch}{c_h}

\title{
Exact Asymptotics for the Random Coding Error Probability%
}%
\author{\authorblockN{Junya Honda}
\authorblockA{Graduate School of Frontier Sciences,
The University of Tokyo\\
 Kashiwa-shi Chiba 277--8561, Japan\\
Email: honda@it.k.u-tokyo.ac.jp}
}

\begin{document}
\maketitle
\allowdisplaybreaks[3]

\begin{abstract}
Error probabilities of random codes for
memoryless channels are considered in this paper\footnote{%
This paper is the full version of \cite{exact_isit}
in ISIT2015 with some corrections and refinements.}.
In the area of communication systems,
admissible error probability
is very small
and it is sometimes more important
to discuss the relative gap
between the achievable error probability and its bound
than to discuss the absolute gap.
%
%
Scarlett et al.~derived a good upper bound
of a random coding union bound based on
the technique of saddlepoint approximation
but it is not proved that
the relative gap of their bound converges to zero.
This paper derives a new bound on the achievable error probability
in this viewpoint for a class of memoryless channels.
The derived bound is strictly
smaller than that by Scarlett et al.~and
its relative gap with the random coding error probability (not a union bound)
vanishes as the block length increases for a fixed coding rate.
\end{abstract}


\section{Introduction}
It is one of the most important task of information theory
to clarify the achievable performance of channel codes under
finite block length.
For this purpose Polyanskiy \cite{second_polyanskiy} and Hayashi \cite{second_hayasi}
considered the achievable coding rate
under a fixed error probability and a block length.
They revealed that the next term to the channel capacity
is $\lo(1/\sqrt{n})$ for the block length $n$
and expressed by a percentile of a normal distribution.

The essential point
for derivation of such a bound
is to evaluate error probabilities
of channel codes with an accurate form.
For this evaluation an asymptotic expansion of sums of random variables
is used in \cite{second_polyanskiy}.
%
On the other hand,
the admissible error probability in communication systems is very small,
say, $10^{-10}$ for example.
In such cases it is sometimes more important to
consider
the {\it relative} gap
between the achievable error probability and its bound
than the absolute gap.
Nevertheless,
an approximation of a tail probability obtained by the asymptotic expansion
sometimes results in a large relative gap
and
it is known that the technique of saddlepoint approximation and
the (higher-order)
large deviation principle is a more powerful tool
rather than 
the asymptotic expansion \cite{saddlepoint}.

Bounds of the error probability of random codes
with a small relative gap
have been researched extensively
although most of them treat a fixed rate $R$
whereas \cite{second_polyanskiy}\cite{second_hayasi}
consider varying rate for the fixed error probability.
Gallager \cite{gallager_map} derived an upper bound called a random coding union bound
on the rate of exponential decay
of the random coding error probability for fixed rate $R$.
It is proved that this exponent of the random code is tight
for both rates below the critical rate \cite{gallager_map}
and above the critical rate \cite{dyachkov}.

There have also been many researches on tight bounds
of the random coding error probability with vanishing or
constant relative error for a fixed rate $R$.
Dobrushin \cite{dobrushin} derived a bound of the random coding error probability
for symmetric channels in the strong sense that
each row and the column of the transition probability matrix
are permutations of the others.
The relative error of this bound is asymptotically bounded by a constant.
In particular, it vanishes in the case that the channel satisfies a nonlattice condition.

For general class of discrete memoryless channels,
Gallager \cite{gallager_tight} derived a bound with a vanishing relative error
for the rate below the critical rate based on the technique of exact asymptotics
for i.i.d.~random variables,
and Altu\u{g} and Wagner \cite{altug_journal} corrected his result
for singular channels.
For general (possibly variable) rate $R$,
Scarlett et al.~\cite{scarlett} derived a simple
upper bound (we write this as $P_{\mathrm{S}}(n)$)
of a random coding union bound $P_{\mathrm{RCU}}(n)$ based on
the technique of saddlepoint approximation
and showed that $P_{\mathrm{RCU}}(n)\le (1+\so(1))P_{\mathrm{S}}(n)$
for nonsingular finite-alphabet discrete memoryless channels \cite{scarlett}.
However,
This bound does not assure $P_{\mathrm{RCU}}(n)= (1+\so(1))P_{\mathrm{S}}(n)$.


In this paper
we consider the error probability $P_{\mathrm{RC}}$ of
random coding for a fixed but arbitrary rate $R$ below the capacity.
We derive a new bound $P_{\mathrm{new}}$ which satisfies
$P_{\mathrm{new}}(n)=(1+\so(1))P_{\mathrm{RC}}(n)$ for
(possibly infinite-alphabet or nondiscrete) nonsingular memoryless channels
such that random variables
associated with the channels satisfy a condition
called a strongly nonlattice condition.
The derived bound matches that by Gallager \cite{gallager_tight}
for the rate below the critical rate\footnote{%
In the ISIT proceedings version
it was described that the result contradicts the bound in \cite{gallager_tight}
but it was the confirmation error of the author because of the difference of notations
between this paper and \cite{gallager}.
See Remark \ref{rm_gallager} for detail.
}.

The essential point to derive the new bound
is that we optimize the parameter depending on
the sent and the received sequences $(\bmX,\bmY)$ to bound the
error probability.
This fact contrasts to discussion in \cite{scarlett} and the classic random coding error exponent
where the parameter is first fixed and optimized
after the expectation over $(\bmX,\bmY)$ is taken.
We confirm that this difference actually affects
the derived bound and
by this difference
we can assure that
the bound also becomes a lower bound of the probability
with a vanishing relative error.



%


\section{Preliminary}
We consider a memoryless channel
with input alphabet $\calX$ and output alphabet $\calY$.
The output distribution for input $x\in\calX$
is denoted by $W(\cdot |x)$.
Let $X\in \calX$ be a random variable with distribution $P_X$
and $Y\in\calY$ be
following
$W(\cdot|X)$ given $X$.
We define $P_Y$ as the marginal distribution of $Y$.
We assume that
$W(\cdot|x)$ is absolutely continuous with respect
to $P_Y$ for any $x$
with density
\begin{align}
\nu(x,y)&=\frac{\rd W(\cdot|x)}{\rd P_Y}(y)\per\n
\end{align}
We also assume that the mutual information is finite, that is,
$I(X;Y)=\E_{XY}[\log \nu(X,Y)]<\infty$.

Let $X'$ be a random variable
with the same distribution as $X$ and independent of $(X,Y)$
and define
$r(x,y,x')=\log \nu(x',y)/\nu(x,y)$.
Since
$\nu(X,Y)\allowbreak >0$
holds
almost surely
we have
$r(X,Y,X')\in\uR=[-\infty,\infty)$ is well-defined almost surely.
$(\bmX,\bmY,\bmX')=((X_1,\cdots,X_n),\,(Y_1,\cdots,Y_n),\,(X_1',\cdots,X_n'))$ denotes
$n$ independent copies of $(X,Y,X')$.
We define $r(\bmX,\bmY,\bmX')=\sum_{i=1}^n r(X_i,Y_i,X_i')$.

We consider the error probability of
a random code such that each element of codewords
$(\bmX_1,\cdots,\allowbreak
\bmX_M)\in \calX^{n\times M}$
is generated independently from distribution $P_X$.
The coding rate of this code is given by
$R=(\log M)/n$.
We use
the maximum likelihood decoding with ties broken uniformly at random.

\subsection{Error Exponent}
Define a random variable $Z(\la)$ on the space of functions
$\mathbb{R}\to\mathbb{R}$ by
\begin{align}
Z(\la)&=
\log\E_{X'}\left[\e^{\la r(X,Y,X')}\right]\n
\end{align}
and its derivatives by
\begin{align}
\bibunZ{m}(\la)&=
\frac{\rd^m}{\rd \la^m}\log\E_{X'}\left[\e^{\la r(X,Y,X')}\right]\com\n
\end{align}
which we sometimes write by $Z'(\la),\,Z''(\la),\cdots$.
Here $\E_{X'}$ denotes the expectation over $X'$
for given $(X,Y)$.
We define\footnote{%
We omit the discussion on the multi-valuedness
of $\log z$.
The discussion involving logarithm of a complex number in this paper
arises by following \cite[Sect.\,XVI.2]{feller_vol2} and refer this to see that
no problem occurs.%
}%
\begin{align}
Z(\la+\im \xi)&=
\log
\E_{X'}\left[\e^{(\la +\im \xi)r(X,Y,X')}\right]
\nn
\cumu(\la+\im \xi)&=
\log
\left|
\E_{X'}\left[\e^{(\la +\im \xi)r(X,Y,X')}\right]
\right|
\com\n
\end{align}
where $\la,\xi\in \mathbb{R}$ and $\im$ is the imaginary unit.
Here we always consider the case $\la>0$ and define $\e^{(\la+\im \xi)(-\infty)}=0$.
We define
\begin{align}
Z_i(\la)=
\log\E_{X'}\left[\e^{\la r(X_i,Y_i,X')}\right],
\;\;
\bar{Z}(\la)=
\frac1n\sum_{i=1}^n Z_i(\la)\per\n
\end{align}
$\cumui{i},\,\cumub,\bibunZ{m}_i$ and $\bar{Z}^{(m)}$ are defined in the same way.

The random coding error exponent
for $0<R<I(X;Y)$
is
denoted by
\begin{align}
E_r(R)
&=
-\inf_{(\alpha,\la)\in [0,1]\times [0,\infty)}\{\alpha R+\log \E[\e^{\alpha Z(\la)}]\}\nn
&=
-\min_{\alpha\in (0,1]} \{\alpha R+\log \E[\e^{\alpha Z(1/(1+\alpha))}]\}\com
\label{def_exponent}
\end{align}
and we write the optimal solution of $(\alpha,\la)$ as
$(\rho,\eta)=(\rho,1/(1+\rho))$.
We write $\log \E[\e^{\alpha Z(1/(1+\alpha))}]=\La(\alpha)$.

In the strict sense
the random coding error exponent represents
the supremum of \eqref{def_exponent} over $P_X$
but for notational simplicity 
we fix $P_X$ and omit its dependence.
See \cite[Theorem 2]{altug_journal} for
a condition that there exists $P_X$ which attains
this supremum.

%

Let $P_{\rho}$ be the probability measure
such that
$\rd P_{\rho}/\rd P=\e^{\rho Z(\eta)-\La(\rho)}$.
We write the expectation under $P_{\rho}$ by
$\E_{\rho}$
and define
\begin{align}
\mu_i&=\E_{\rho}[\bibunZ{i}(\eta)]
=\e^{-\La(\rho)}\E[\bibunZ{i}(\eta)\e^{\rho Z(\eta)}]\nn
\si_{ij}&=\E_{\rho}[(\bibunZ{i}(\eta)-\mu_i)(\bibunZ{j}(\eta)-\mu_j)]\nn
&=\e^{-\La(\rho)}\E[(\bibunZ{i}(\eta)-\mu_i)(\bibunZ{j}(\eta)-\mu_j)\e^{\rho Z(\eta)}]\nn
\Sigma_{ij}&=\left(
\begin{array}{cc}
\si_{ii}&\si_{ij}\\
\si_{ji}&\si_{jj}
\end{array}
\right)\per
\n
\end{align}
From derivatives of
$\alpha R+\log \E[\e^{\alpha Z(\la)}]$ in $\alpha$ and $\la$ we have
\begin{align}
\hen{\log \E[\e^{\alpha Z(\eta)}]}{\alpha}\bigg|_{\alpha=\rho}
&=\mu_0
\begin{cases}
=-R,& \mbox{if } R\ge R_{\mathrm{crit}},\\
<-R,& \mbox{otherwise},
\end{cases}
\label{la_bibun}\\
\hen{\log \E[\e^{\rho Z(\la)}]}{\lambda}\bigg|_{\lambda=\eta}
&=
\al\mu_1=0\per\label{z_ave}
\end{align}
where $R_{\mathrm{crit}}$ is the critical rate, that is,
the largest $R$ such that the optimal solution of \eqref{def_exponent}
is $\rho=1$.
We assume that $\mu_2>0$, or equivalently,
$P_Y[|\mathcal{Q}(Y)\setminus\{0\}|>1]>0$
where $\mathcal{Q}(y)$ is the support of
$\nu(X',y)$.
This corresponds to
the non-singular assumption in \cite{scarlett}\cite{altug}
for the finite alphabet.

To avoid somewhat technical argument on the continuity and integrability
we also assume
that there exists $\alpha,\gapxia>0$ and a neighborhood $\mathcal{S}$ of $\la=\eta$
such that for any $0<\gapxi<b_2<2\pi/h\le \infty$
\begin{align}
&\sup_{\la\in \mathcal{S}}\E_{\rho}[\e^{\alpha |\bibunZ{m}(\la)|}]<\infty\com\quad i=1,2,3,\nn
&\sup_{\la\in \mathcal{S},\,\xi\in[-\gapxia,\gapxia]}
\E_{\rho}[\e^{\alpha |(\partial^4/\partial \xi^4)Z(\la+\im\xi)|}]<\infty\com\nn
&\sup_{\la\in \mathcal{S},\,\xi\in [\gapxi,b_2]}
\E_{\rho}[\e^{\alpha |\cumu(\la+\im\xi)-\cumu(\la)|}]<\infty\per
\label{regularity_moment}
\end{align}
where $h\ge 0$ is given later.
Note that these conditions trivially hold 
if the input and output alphabets are finite.

\subsection{Lattice and Nonlattice Distributions}
In the asymptotic expansion with an order higher than
the central-limit theorem,
it is necessary to consider cases that
the distribution is lattice or nonlattice separately.
Here we call that a random variable $V\in\mathbb{R}^m$ has a lattice distribution
if $V\in\{ a+\sum_{i=1}^m b_ih_i:\{b_i\}\in\mathbb{Z}^m\}$ almost surely
for some $a\in\mathbb{R}^m$ and linearly independent vectors
$\{h_i\}_{i=1}^m\in\mathbb{R}^{m\times m}$.
For the case $m=1$
we call the largest $h_1$ satisfying the above condition the span of the lattice.


On the other hand, we call that $V\in\mathbb{R}^m$ has a strongly nonlattice distribution
if $|\E[\e^{\im \inner{\xi,V}}]|<1$
for all $\xi\in\mathbb{R}^m\setminus\{0\}$, where $\inner{\cdot,\cdot}$ denotes the inner product.
Note that a one dimensional random variable $V\in \mathbb{R}$ is
lattice or strongly nonlattice but, in general, there exists a
random variable which is not lattice and not strongly nonlattice.

As given above,
a lattice distribution is defined for
a random variable $V\in\mathbb{R}^{m}$ in standard references such as \cite{ranga}.
In this paper we call that the distribution of $V\in\underline{\mathbb{R}}$ is lattice
if the conditional distribution of $V$ given $V>-\infty$ is lattice
and nonlattice otherwise.
It is easy to see that no contradiction occurs under this definition.

We consider the following condition regarding lattice and nonlattice distributions.
\begin{definition}\label{def_lattice}{\rm
We call that
the log-likelihood ratio $\nu$ satisfies the lattice condition with span $h>0$
if
the conditional distribution of $\log \nu(X,Y)$ given $Y$
is lattice with span $h m_Y$ almost surely
where $m_Y\in \mathbb{N}$ may depend on $Y$
and $h$ is the largest value satisfying this condition.
}\end{definition}
For notational simplicity
we define the span of the lattice for $\nu$
to be $h=0$
if $\nu$ does not satisfy the lattice condition.
Other than the classification of $\nu$,
we also discuss cases that
$(Z(\eta),Z'(\eta))$ is strongly nonlattice or not separately.

Note that a one-dimentional random variable $V\in\mathbb{R}$ with support $\mathrm{supp}(V)$
is always lattice if $|\mathrm{supp}(V)|\le 2$,
and is strongly nonlattice except for some special cases
if $|\mathrm{supp}(V)|\ge 3$.
Similarly,
a two-dimensional random variable $V\in\mathbb{R}^2$
is always not strongly nonlattice if
$|\mathrm{supp}(V)|\le 3$,
and is strongly nonlattice except for some special cases
if $|\mathrm{supp}(V)|\ge 4$.
Based on this observation
we see that most channels with input and output alphabet sizes larger than 3
are strongly nonlattice.
Another example of each class of channels
(excluding those with specially chosen parameters)
are given in Table.\,\ref{table_class}.

\begin{table*}[b]
\begin{center}%
\caption{Classification of Nonsingular Channels.}%
\label{table_class}%
\begin{tabular}{l l| c| c|}%
&&\multicolumn{2}{c|}{$(Z(\eta),Z'(\eta))$}\\
&&not strongly nonlattice&strongly nonlattice\\
\hline
\multirow{2}{*}{log-likelihood ratio $\nu$}&lattice&BSC&asymmetric BEC\\
\cline{2-4}
&nonlattice&ternary symmetric channels&binary asymmetric channels\\
\hline
\end{tabular}%
\end{center}
\end{table*}

\begin{remark}{\rm
The above conditions are different from the condition
considered in \cite{scarlett} as a classification of lattice and nonlattice cases.
This difference arises from two reasons.
First, we consider $Z'(\eta)$ in addition to $Z(\eta)$ to derive an accurate bound.
Second,
the proof of \cite[Lemma 1]{scarlett}
does not use the correct span when applying the result \cite[Sect.\,VII.1, Thm.\,2]{petrov_sums}.
}\end{remark}
%




\section{Main Result}
Define
\begin{align}
g_h(u)=
1-\frac{\e^{-\frac{h\eta}{\e^{h\eta}-1}u}(1-\e^{-h\eta u})}{h\eta u}.\n
\end{align}
for $h\ge0$.
Here we define
$(\e^{x}-1)/x=(1-\e^{-x})/x=1$
for $x=0$
and therefore
$g_0(u)=\lim_{h\downarrow 0}g(u)
=
1-\e^{-u}$.
We give some properties on $g_h$
in Appendix \ref{append_gh}.
Now
we can represent the random coding error probability
as follows.

\begin{theorem}\label{thm_expansion1}
Fix any $0< R<I(X;Y)$ and $\ep>0$,
and let
$\de_2>0$ be sufficiently small.
Then,
for the span $h\ge 0$ of the lattice for $\nu$,
there exists $n_0>0$ such that for all $n\ge n_0$
\begin{align}
\lefteqn{
(1-\ep)
\E\!\left[
g_h\left(
(1-\ep)\frac{
 \e^{n(\bar{Z}(\eta)+R-(\bar{Z}'(\eta))^2/2(\mu_2-\de_2))}
}{\eta\sqrt{2\pi n \mu_2}}
\right)
\right]
}\nn
&\le
P_{\mathrm{RC}}(n)\nn
&\le
(1+\ep)
\E\!\left[
g_h\left(
(1+\ep)\frac{
 \e^{n(\bar{Z}(\eta)+R-(\bar{Z}'(\eta))^2/2(\mu_2+\de_2))}
}{\eta\sqrt{2\pi n \mu_2}}
\right)
\right]\com
\n
\end{align}
\end{theorem}
By this theorem we can reduce the evaluation of error probability
into that of an expectation over two-dimensional random variable
$(\bar{Z}(\eta),\bar{Z}'(\eta))$, although this expectation is still difficult
to compute.
If $(Z(\eta),Z'(\eta))$ is strongly nonlattice then
we can derive the following bound which gives an explicit representation
for the asymptotic behavior of $P_{\mathrm{RC}}$.
\begin{theorem}\label{thm_main}
Fix $0<R<I(X;Y)$ and assume that
$(Z(\eta),Z'(\eta))$ has a strongly nonlattice distribution.
Then
\begin{align}
\lefteqn{
P_{\mathrm{RC}}(n)
}\nn
&=
\begin{cases}
\frac{\psi_{\rho,h}\mu_2^{(1-\rho)/2}(1+\so(1))}%
{\eta^{\rho}(2\pi n)^{(1+\rho)/2}\sqrt{(\mu_2\si_{00}+\rho|\Sigma_{01}|)}}
\e^{-nE_r(R)},
&R> R_{\mathrm{crit}},\\
\frac{h(1+\so(1))}{2(\e^{\eta h}-1)\sqrt{2\pi n(\mu_2+\si_{11})}}
\e^{-nE_r(R)},
&R=R_{\mathrm{crit}},\\
\frac{h(1+\so(1))}{(\e^{\eta h}-1)\sqrt{2\pi n(\mu_2+\si_{11})}}
\e^{-nE_r(R)},
&R<R_{\mathrm{crit}},
\end{cases}\label{thm_align}
\end{align}
where
\begin{align}
\psi_{\rho,h}
&=
\int_{-\infty}^{\infty} \e^{-\rho w}g_h(\e^w)\rd w\nn
&=\frac{\Gamma(1-\rho)}{\rho}
\left(\frac{h\eta}{\e^{h\eta}-1}\right)^{\rho+1}
\frac{\e^{h}-1}{h}\n
\end{align}
for the gamma function $\Gamma$.
\end{theorem}
We prove Theorems \ref{thm_expansion1} and \ref{thm_main} in
Sections \ref{sec_decomp} and \ref{sec_second}, respectively.
From this theorem we see that
at least for the strongly nonlattice case
the error probability of the random coding is
\begin{align}
P_{\mathrm{RC}}(n)&=
\begin{cases}
\Omega(n^{-(1+\rho)/2}\e^{-nE_r(R)}),&R>R_{\mathrm{crit}}\\
\Omega(n^{-1/2}\e^{-nE_r(R)}),&R\le R_{\mathrm{crit}}.
\end{cases}
\label{order}
\end{align}
The RHS of \eqref{order} for $R>R_{\mathrm{crit}}$ is the same expression
as
the upper bounds in \cite{scarlett}\cite{altug}
but our bound is tighter in its coefficient
and is also assured to be the lower bound.

It may be possible to
derive a similar bound as Theorem \ref{thm_main}
for the case that $(Z(\eta),Z'(\eta))$ is not strongly nonlattice
by replacement of integrals with summations,
but for this case the author was not able to find an
expression of the asymptotic expansion straightforwardly applicable
to our problem and this remains as a future work.

\begin{remark}\label{remark_union}{\rm
We can show in the same way as Theorem \ref{thm_main} that
the random coding {\it union} bound
is obtained
by replacement of $\psi_{\rho,h}$
with
\begin{align}
\lefteqn{
\!\!\!\!\!
\int_{-\infty}^{\infty}\e^{-\rho w}
\min\left\{\frac{h\eta \e^w}{\e^{h\eta}-1},1\right\}\rd w
}\nn
&=
\left(\frac{1}{1-\rho}+\frac{1}{\rho}\right)
\left(\frac{h\eta}{\e^{h\eta}-1}\right)^{\rho}.\n
\end{align}
On the other hand,
the terms $|\rho\Sigma_{01}|$ and $\si_{11}$
in the square roots of \eqref{thm_align} are the
characteristic parts of the analysis of this paper
obtained by the optimization of parameter $\lambda$
depending on $(\bmX,\bmY)$.
Thus, the optimization of $\lambda$
is necessary to derive a tight coefficient
whether we evaluate the error probability itself
or the union bound.
}\end{remark}


\begin{remark}{\rm
The results in this paper
assume a {\it fixed} coding rate $R$
and are weaker in this sense than the result by Scarlett et al.~\cite{scarlett}
where they assure
an upper bound for varying rate
by leaving an integral (or a summation) to a form such that
the integrant depends on $n$.
It may be possible to extend Theorem \ref{thm_expansion1} for 
varying rate since the most part of the proof deals
with $R$ and the error probability of each codeword
separately.
However, the proof of Theorem \ref{thm_main}
heavily depends on fixed $R$ and
it is also an important problem to derive an easily computable bound
for varying rate.
}\end{remark}

\begin{remark}\label{rm_gallager}{\rm
In \cite{gallager_tight} it is shown
for discrete nonlattice\footnote{%
There is a calculation error for the lattice case
in \cite{gallager_tight}
with a redundant factor $\sqrt{\pi}$.
} channels with $R<R_{\mathrm{crit}}$ that
\begin{align}
P_{\mathrm{RC}}(n)
&=
\frac{(1+\so(1))}{\eta\sqrt{2\pi n\mu_2'}}
\e^{-nE_r(R)},\label{result_gallager}
\end{align}
where
\begin{align}
\mu_2'&=
\hen{^2\log \E[\e^{Z(\la)}]}{\la^2}\bigg|_{\la=\eta}\nn
&=
\frac{2\sum_y(\omega_0(y)\omega_2(y)-\omega_1(y)^2)}%
{\sum_{y}\omega_0^2(y)}
\end{align}
for
\begin{align}
\omega_m(y)=\sum_x P_X(x)(\log W(y|x))^m\sqrt{W(y|x)}\per\n
\end{align}
The author misunderstood that $\mu_2'=\mu_2$ in the ISIT version
and described that Theorem \ref{thm_main} contradicts \eqref{result_gallager}.
The correct calculation show that
$\mu_2'\neq \mu_2$ and
\begin{align}
\mu_2=\sigma_{11}&=
\frac{\sum_{y}
\left(\omega_0(y)\omega_2(y)-\omega_1(y)^2
\right)
}{\sum_{y}\omega_0^2(y)}\n
\end{align}
for $(\rho,\,\eta)=(1,\,1/2)$.
Therefore no contradiction occurs
between this paper and \cite{gallager_tight}.
}\end{remark}

\section{First Asymptotic Expansion}\label{sec_decomp}
In this section
we give a sketch of the proof of Theorem \ref{thm_expansion1}.
We prove Theorem \ref{thm_expansion1} separately depending on
whether $\nu$ satisfies the lattice condition or not.
The proofs are different to each other in some places
for two reasons.
First, we cannot ignore the case that a codeword has the same likelihood
as that of the sent codeword under the lattice condition
whereas such a case is almost negligible in the nonlattice case.
Second, especially in the case of infinite alphabet
we have to use the asymptotic expansion
with a careful attention to components implicitly assumed to be fixed
and the derivation of asymptotic expansion varies in some places
between the lattice and nonlattice cases
regarding this aspect.


Here we give a proof of Theorem \ref{thm_expansion1}
for the case that $\nu$ satisfies the lattice condition
with span $h>0$.
The proof for the nonlattice case is easier than the lattice case
in most places because ties of likelihoods can be almost ignored
as described above.
See Appendix \ref{append_nonlattice} for the difference
of the proof in the nonlattice case.

Now define
\begin{align}
\pz(\bmx,\bmy)&=P_{\bmX'}[r(\bmx,\bmy,\bmX')=0]\nn
\pp(\bmx,\bmy)&=P_{\bmX'}[r(\bmx,\bmy,\bmX')>0]
=P_{\bmX'}[r(\bmx,\bmy,\bmX')\ge h]\per\label{eq_offset}
\end{align}
The last equation of
\eqref{eq_offset} holds
since
$r(x,y,x')=\log \nu(x',y)-\log\nu(x,y)$
and the offset of the lattice of
$\log \nu(x',y)$ equals to that of $\log\nu(x,y)$ given $y$.
Under the maximum likelihood decoding,
the average error probability $P_{\mathrm{RC}}$ is expressed as
$P_{\mathrm{RC}}=
\E_{\bmX\bmY}[
q_{M}(\pp(\bmX,\bmY),\pz(\bmX,\bmY))]$
for
\begin{align}
q_M(\pp,\pz)&=
1-(1-\pp)^{M-1}
\nn
&\!\!\!\!\!\!\!\!\!\!\!\!\!\!\!
+
\sum_{i=1}^{M-1}\pz^i(1-\pp-\pz)^{M-i-1}{{M-1}\choose i}\left(1-\frac{1}{i+1}\right)\!.
\label{error_moto}
\end{align}
Here the first term corresponds to the probability that the likelihood of some codeword
exceeds that of the sent codeword,
and each component of the second term corresponds to the probability that
$i$ codewords have the same likelihood as the sent codeword
and the others do not exceed this likelihood.

One of the most basic bound for this quantity is to use a union bound
given by
\begin{align}
q_M(\pp,\pz)\le
\min\{1,(M-1)(\pp+\pz)\}\per\n
\end{align}
A lower can also be found in, e.g., \cite[Chap.~23]{poly_lecture}.
For evaluation of the error probability with a vanishing relative error
the following lemma is useful.
\begin{lemma}\label{thm_unify}
It holds for any $c\in(0,1/2)$ that
\begin{align}
\lefteqn{
\blimsup_{M\to\infty}
\sup_{(\pp,\pz)\in (0,1/3]^2 : \pp\le M^{c}\pz}
\frac{q_M(\pp,\pz)}{1-\frac{\e^{-M\pp}(1-\e^{-M\pz})}{M\pz}}
}\nn
&=
\bliminf_{M\to\infty}
\inf_{(\pp,\pz)\in (0,1/3]^2 : \pp\le M^{c}\pz}
\frac{q_M(\pp,\pz)}{1-\frac{\e^{-M\pp}(1-\e^{-M\pz})}{M\pz}}=1\per\n
\end{align}
\end{lemma}
We prove this lemma in Appendix \ref{proof_unify}.
We see from this theorem that
the error probability can be approximated by
\begin{align}
1-\frac{\e^{-M\pp(\bmX,\bmY)}(1-\e^{-M\pz(\bmX,\bmY)})}{M\pz(\bmX,\bmY)}\n
\end{align}
for $(\bmX,\bmY)$ satisfying some regularity condition.

Next we consider the evaluation of $\pz(\bmX,\bmY)$ and $\pp(\bmX,\bmY)$.
We use Lemma \ref{thm_tool_lattice} in the following as a fundamental tool of the proof.
Let $V_1,\cdots,V_n\in\underline{\mathbb{R}}$ be (possibly not identically distributed)
independent lattice random variables
such that the greatest common divisor of their spans\footnote{
The greatest common divisor for a set $\{h_1,h_2,\cdots\},\,h_i>0$,
is defined as $h>0$ if $h$ is the maximum number such that
$h_i/h\in\mathbb{N}$ for all $i$ and defined as $0$ if
such $h$ does not exist.
} is $h$.
Define
\begin{align}
\La_{V_i}(\la)=\log \E[\e^{\la V_i}]\com
\;\;
\La_{\bm{V}}(\la)=\sum_{i=1}^n\La_{V_i}(\la)\per\n
\end{align}
Then its large deviation probability is
evaluated as follows.
\begin{lemma}\label{thm_tool_lattice}
Fix $x>\sum_{i=1}^n\E[V_i]$ such that
$\Pr[(V_i-x)/h\in \mathbb{Z}]=1$
and define
$\la^*>0$ as the solution of $\La_{\bmV}'(\la^*)=x$.
Let $\ep,\dephi,\gapxia \,\ua_2,\oa_2,\oa_4>0$ and $\ua_3,\oa_3\in\mathbb{R}$ be arbitrary.
Then there exists
$\gapxi=\gapxi(\gapxia, \ua_2,\oa_2,\ua_3,\oa_3,\oa_4),
n_0=n_0(\ep,\gapxia,\dephi, \ua_2,\oa_2,\ua_3,\oa_3,\oa_4)>0$
such that
\begin{align}
\left|
\frac{\Pr[\sum_{i=1}^n V_i=x]}{
\frac{
h\e^{-n(\eta x-\La_{\bmV}(\la^*))}
}{\sqrt{2\pi\La_{\bmV}''(\la^*)}}
}-1
\right|
&\le
\ep\com\nn
\left|
\frac{\Pr[\sum_{i=1}^n V_i\ge x+h]}{
\frac{
h\e^{-n(\eta x-\La_{\bmV}(\la^*))}
}{(\e^{h\la^*}-1)\sqrt{2\pi\La_{\bmV}''(\la^*)}}
}-1
\right|
&\le
\ep\com
\n
\end{align}
hold
for all $n\ge n_0$ satisfying
\begin{align}
&n \ua_m\le\sum_{i=1}^n\bibun{^m\La_{V_i}(\la)}{\la^{m}}\bigg|_{\la=\la^*}
\le n\oa_m,\qquad i=2,3,\nn
&\sum_{i=1}^n\left|\hen{^4\La_{V_i}(\la^*+\im\xi)}{\xi^{4}}\right|
\le n\oa_4,
\qquad \forall |\xi|\le \gapxia\nn
&\sum_{i=1}^n\left(\log |\E[\e^{(\la^*+\im \xi)V_i}]|
-\log \E[\e^{\la^*V_i}]\right)
\le -n \dephi,\nn
&\phantom{wwwwwwwwwwwwwww}\forall \xi\in[-\pi/h,\pi/h]
\setminus[-\gapxi,\gapxi]\per\n
\end{align}
\end{lemma}
The proof of this lemma is largely the same as that
of \cite[Thm.\,3.7.4]{LDP} for the i.i.d.~case
and given in Appendix \ref{append_tool_lattice}.

Let $\gapxia,\gap_1,\gap_2,\gap_3,\gapla,\dephi,\oa_4>0$
satisfy $\gap_2<\min\{\mu_2/2,\allowbreak \mu_2\sqrt{R/12}\}$.
To apply Lemma \ref{thm_tool_lattice} we consider the
following sets $\mathcal{A}_m,\,m=2,3,\,\mathcal{B},\mathcal{C}$
to formulate regularity conditions.
\begin{align}
\mathcal{A}_m&=
\left\{
f_1\in\mathcal{C}_1: \forall \la,\: \left| f_m(\la)-\der{m}\right|\le \gap_2
\right\}\com\nn
\mathcal{B}&=
\left\{
f_2\in \mathcal{C}_2: \forall \la,\xi \notin [-\gapxi,\gapxi],\:f_2(\la,\xi)\le -\dephi
\right\}
\com\nn
\mathcal{C}&=
\left\{
f_2\in \mathcal{C}_2: \forall \la,\xi\in[-\gapxia,\gapxia],\:f_2(\la,\xi)\le \oa_4
\right\}
\com\n
\end{align}
where $\mathcal{C}_1$ and $\mathcal{C}_2$ are the spaces of continuous functions
$[\eta-\gapla,\eta+\gapla]\to\mathbb{R}$
and
$[\eta-\gapla,\eta+\gapla]\times
[-\pi/h,\pi/h]
\to\mathbb{R}$,
respectively,
and $\gapxi$ is a constant determined from
$\gapxia,\ua_2,\oa_2,\ua_3,\oa_3,\oa_4$ with
Lemma \ref{thm_tool_lattice}.

We define the event $S$ as
\begin{align}
S&=
\{|\bar{Z}^{(1)}(\eta)|\le \gap_1\}\cup\{\bar{Z}^{(2)}(\la)\in \mathcal{A}_2\}
\cup\{\bar{Z}^{(3)}(\la)\in \mathcal{A}_3\}
\nn
&\quad\cup
\{\cumub(\la+\im\xi)-\cumub(\la)\in\mathcal{B}\}\nn
&\quad\cup
\left\{\left|
\hen{^4}{\xi^4}
\bar{Z}^{(4)}(\la+\im\xi)
\right|\in\mathcal{C}\right\}
\com\n
\end{align}
where we regard $\bar{Z}(\la+\im\xi)$ as
function $(\la,\xi) \mapsto \allowbreak\bar{Z}(\la+\im\xi)$.
Under this condition we can bound the excess probability of the likelihood of each codeword
given the sent codeword $\bmX$ and the received sequence $\bmY$ as follows.
\begin{lemma}\label{lem_f}
Let
$\ep>0$ be arbitrary and $\gap_1>0$
in the definition of $S$
be sufficiently small with respect
to $\gapla$.
Then, there exists $n_1>0$ such that
under the event $S$ it holds
for all $n\ge n_1$ that,
\begin{align}
\lefteqn{
\!\!\!\!\!
\frac{
h\e^{n(\bar{Z}(\eta)-\bar{Z}'(\eta)^2/2(\der{2}-\gap_2))}
}{\sqrt{2\pi n (\der{2}+\gap_2)}}(1-\ep)
\le
\pz(\bmX,\bmY)
}\nn
&\qquad\qquad\le
\frac{
 h\e^{n(\bar{Z}(\eta)-\bar{Z}'(\eta)^2/2(\der{2}+\gap_2))}
}{\sqrt{2\pi n (\der{2}-\gap_2)}}
(1+\ep)\com
\displaybreak[3]
\nn
\lefteqn{
\!\!\!\!\!
\frac{
h\e^{n(\bar{Z}(\eta)-\bar{Z}'(\eta)^2/2(\der{2}-\gap_2))}
}{(\e^{h(\eta+\gapla)}-1)\sqrt{2\pi n (\der{2}+\gap_2)}}(1-\ep)
\le
\pp(\bmX,\bmY)
}\nn
&\qquad\qquad\le
\frac{
 h\e^{n(\bar{Z}(\eta)-\bar{Z}'(\eta)^2/2(\der{2}+\gap_2))}
}{(\e^{h(\eta-\gapla)}-1)\sqrt{2\pi n (\der{2}-\gap_2)}}
(1+\ep)\per\n
\end{align}
\end{lemma}
\begin{proof}
Note that $|\bar{Z}'(\eta)|\le \gap_1$
and $\bar{Z}''(\la)\ge\mu_2/2$ for all $\la\in[\eta-\gapla,\eta+\gapla]$
from $\bar{Z}^{(m)}(\la)\in\mathcal{A}_m$ and \eqref{z_ave}.
From the convexity of $\bar{Z}(\la)$ in $\la$,
if we set $\gap_1\le \gapla \mu_2/2$ then
$\bar{Z}(\la)$ is minimized at a point in $[\eta-\gapla,\eta+\gapla]$
with
\begin{align}
\bar{Z}(\eta)-\frac{(\bar{Z}'(\eta))^2}{2 (\der{2}-\gap_2)}\le
\min_{\la}\bar{Z}(\la)\le \bar{Z}(\eta)-\frac{(\bar{Z}'(\eta))^2}{2 (\der{2}+\gap_2)}\per\n
\end{align}
Thus the lemma follows from Lemma \ref{thm_tool_lattice}.
\end{proof}

Next we define
\begin{align}
\tg{-}(\bmX,\bmY)&=
(1-\ep/2)g_h\left(
\frac{
\e^{n(\bar{Z}(\eta)+R-(\bar{Z}'(\eta))^2/2(\der{2}-\gap_2))}
}{\cs{-}\sqrt{n}}
\right)\com\nn
\tg{+}(\bmX,\bmY)
&=
(1+\ep/2)g_h\left(
\frac{
 \e^{n(\bar{Z}(\eta)+R-(\bar{Z}'(\eta))^2/2(\der{2}+\gap_2))}
}{\cs{+}\sqrt{n}}
\right)\com\nn
\tG{s}&=\E[\tg{s}(\bmX,\bmY)],\,s\in\{-,+\}\com\n
\end{align}
where
\begin{align}
\cs{-}=
\frac{\eta(\e^{h(\eta+\gapla)}-1)\sqrt{2\pi (\mu_2+\gap_2)}}{(\e^{h\eta}-1)(1-\ep/2)}\com\nn
\cs{+}=
\frac{\eta(\e^{h(\eta-\gapla)}-1)\sqrt{2\pi (\mu_2-\gap_2)}}{(\e^{h\eta}-1)(1+\ep/2)}\per\n
\end{align}
Then the error probability can be evaluated as follows.
\begin{lemma}\label{lem_g}
Fix the coding rate $R$ and assume that
the same condition as Lemma \ref{lem_f} holds.
Then, for all sufficiently large $n$,
\begin{align}
\tg{-}(\bmX,\bmY)
\le
q_M(\pp(\bmX,\bmY),\pz(\bmX,\bmY))
\le
\tg{+}(\bmX,\bmY)\per\n
\end{align}
\end{lemma}
This lemma is straightforward from
Lemmas \ref{thm_unify} and \ref{lem_f}.
We use the following lemma to evaluate the contribution
of the case $S^c$.
\begin{lemma}\label{lem_regularity}
Let $\ug(\bmX,\bmY)=\e^{n\rho(\bar{Z}(\eta)+R)}$.
Then
\begin{align}
q_M(\pp(\bmX,\bmY),\pz(\bmX,\bmY))
&\le
\ug(\bmX,\bmY)\com\label{regularity_kantan1}\\
\tg{-}(\bmX,\bmY)
&\le
\frac{1+h\eta/2}{(\cs{-})^{\rho}}\ug(\bmX,\bmY)\per
\label{regularity_kantan2}
\end{align}
Furthermore,
for sufficiently large $\oa_4$
and sufficiently small $\gapla\ll \min\{\gap_2,\gap_3\}$ and $\dephi\ll \gapxi$ we have
\begin{align}
\blimsup_{n\to\infty}
\frac1n\log \E_{\bmX\bmY}[\idx{S^c}\ug(\bmX,\bmY)]<-E_r(R)\per
\n
\end{align}
\end{lemma}
We prove this lemma in Appendix
\ref{append_regularity}.
The proof is obtained by
Cram\'er's theorem for general topological vector spaces \cite[Theorem 6.1.3]{LDP}
with the fact that $\mathcal{C}_1$ and $\mathcal{C}_2$
are separable Banach spaces under the max norm.

\begin{proof}[Proof of Theorem \ref{thm_expansion1}]
From Lemma \ref{lem_g},
it holds for $\gap_1\ll \gapla\ll \min\{\gap_2,\gap_3\},\,\dephi\ll \gapxi$
and sufficiently large $n$
that
\begin{align}
P_{\mathrm{RC}}
&=
\E_{\bmX\bmY}[\idx{S}q_M(\pp(\bmX,\bmY),\pz(\bmX,\bmY))]
\nn
&\quad+
\E_{\bmX\bmY}[\idx{S^c}q_M(\pp(\bmX,\bmY),\pz(\bmX,\bmY))]\nn
&\le
\tG{+}
+
\E_{\bmX\bmY}[\idx{S^c}q_M(\pp(\bmX,\bmY),\pz(\bmX,\bmY))]\per\n
\end{align}
Thus we obtain from Lemma \ref{lem_regularity} that
\begin{align}
\frac{P_\mathrm{RC}}{\tG{+}}
&=
1+\frac{P_\mathrm{RC}-\tG{+}}{\tG{+}}
\le
1+\frac{\E_{\bmX\bmY}[\idx{S^c}\ug(\bmX,\bmY)]}{\tG{+}}\per\n
\end{align}

Similarly we have
\begin{align}
P_{\mathrm{RC}}
&\ge
\E_{\bmX\bmY}[\idx{S}\tg{-}(\bmX,\bmY)]\nn
&=
\tG{-}
-
\E[\idx{S^c}\tg{-}(\bmX,\bmY)]\n
\end{align}
and therefore
\begin{align}
\frac{P_\mathrm{RC}}{\tG{-}}
&\ge
1-\frac{1+h\eta/2}{(\cs{-})^{\rho}}\frac{\ug(\bmX,\bmY)}{\tG{-}}
\n
\end{align}
and we see from Lemma \ref{lem_regularity} and Lemma \ref{lem_main} below
that
\begin{align}
\frac{\ug(\bmX,\bmY)}{\tG{s}}=\so(1),\,s\in\{+,-\}\n
\end{align}
and we obtain Theorem \ref{thm_expansion1}.
\end{proof}

\section{Second Asymptotic Expansion}\label{sec_second}
To prove Theorem \ref{thm_main}
it is necessary to evaluate the expectation $\tG{s}=\E[\tg{s}(\bmX,\bmY)]$.
This expectation can be bounded by Lemma \ref{lem_main} below
and we give a sketch of its proof in this section.

\begin{lemma}\label{lem_main}
Fix the coding rate $0<R<I(X;Y)$
assume that
$(Z(\eta),Z'(\eta))$ is strongly nonlattice.
Then, for any fixed $c_1,c_2>0$ and sufficiently large $n$,
\begin{align}
\lefteqn{
\E\!\left[
g_h\left(
\frac{
 \e^{n(\bar{Z}(\eta)+R-(\bar{Z}'(\eta))^2/2c_1)}
}{c_2\sqrt{n}}
\right)
\right]
}\nn
&=
\begin{cases}
\frac{\psi_{\rho}(c_2\sqrt{n})^{-\rho}}{\sqrt{2\pi n(\si_{00}+\rho|\Sigma_{01}|/c_1)}}
\e^{-nE_r(R)}
(1+\so(1)),
&R> R_{\mathrm{crit}},\\
\frac{h\eta(c_2\sqrt{n})^{-1}}{2(\e^{h\eta}-1)\sqrt{1+\si_{11}/c_1}}
\e^{-nE_r(R)}
(1+\so(1)),
&R=R_{\mathrm{crit}},\\
\frac{h\eta(c_2\sqrt{n})^{-1}}{(\e^{h\eta}-1)\sqrt{1+\si_{11}/c_1}}
\e^{-nE_r(R)}
(1+\so(1)),
&R<R_{\mathrm{crit}}.
\end{cases}\n
\end{align}
\end{lemma}

Let $\Phi_{\Sigma}$ and $\phi_{\Sigma}$
be the cumulative distribution function and the density
of a normal distribution
with mean zero and covariance $\Sigma$, respectively.
We define the $\de$-ball $B_{\de}(z)\in \mathbb{R}^2$ around $z\in\mathbb{R}^2$ as
$B_{\de}(z)=\{z': \Vert z-z'\Vert\le \de\}$.
The oscillation $\omega_f$ of $f$ is defined as
\begin{align}
\omega_f(S)&=
\sup_{z'\in S}f(z')-\inf_{z'\in S}f(z')\com\qquad S\subset \mathbb{R}^2\com\nn
\omega_f(\de;\Phi_{\Sigma})
&=
\sup_{a\in\mathbb{R}^2}
\int \omega_f(B_{\de}(z))\phi_{\Sigma}(z+a)\rd z\per\n
\end{align}

We use the following proposition on the asymptotic expansion for the proof of Lemma \ref{lem_main}.
\begin{proposition}[{\cite[Theorem 20.8]{ranga}}]\label{prop_ranga}
Let $V_1,V_2,\cdots\in \mathbb{R}^2$ be i.i.d.\,strongly nonlattice random variables
with mean zero and covariance matrix $\Sigma$.
Then, there exists a three-degree
polynomial\,\footnote{The explicit representation of $h(z)$ is given in the original
reference \cite{ranga} but we do not use it in this paper.} $h(z)=h(z_1,z_2)$
such that for any function $f(z)$
\begin{align}
&
\bigg|\int f(z)\left(1-\frac{h(z)}{\sqrt{n}}\right)
\phi_{\Sigma}(z)\rd z
-\E[f(\bar{V})]\bigg|\nn
&\qquad
\le
\omega_f(\mathbb{R}^2)\de_n+\omega_f(\de_n; \Phi_{\Sigma})\com\n
\end{align}
where
$\de_n$ satisfies $\lim_{n\to\infty}\sqrt{n}\de_n=0$
and does not depend on $f$.
\end{proposition}

To apply this proposition
we define
\begin{align}
f_n(z)=
\e^{-\sqrt{n}\rho z_1}
g_h\left(
\frac{
 \e^{\sqrt{n}z_1-z^2/2c_1}
}{c_2\sqrt{n}}
\right)\per\n
\end{align}
The oscillations $\omega_{f_n}(\mathbb{R}^2)$ and
$\omega_{f_n}(\de_n;\Phi)$ of $f_n$ are equal to
those of
\begin{align}
\e^{-\sqrt{n}\rho(z_1-\sqrt{n}\De)}
g_h\left(
\frac{
 \e^{\sqrt{n}(z_1-\sqrt{n}\De)-z^2/2c_1)}
}{c_2\sqrt{n}}
\right)\n
\end{align}
from their definitions.

We can bound the oscillation of $f_n$ as follows.
\begin{lemma}\label{lem_osci}
It holds that
\begin{align}
\omega_{f_n}(\mathbb{R}^2)
&=
\lo(n^{-\rho/2})\com
\label{osci1}\\
\omega_{f_n}(\de_n;\Phi)
&=
\so(n^{-\rho/2})\per
\label{osci2}
\end{align}
Furthermore, if $\rho<1$ then
\begin{align}
\omega_f(\de_n;\Phi)
&=
\so(n^{-(1+\rho)/2})
\per
\label{osci3}
\end{align}
\end{lemma}
We prove this lemma in Appendix \ref{append_osci}.
By this lemma we can apply Proposition \ref{prop_ranga} to the proof of Lemma \ref{lem_main},
which we give in Appendix \ref{append_second}.

\section{Conclusion}
We derived a bound of random coding error probability,
the relative gap of which converges to zero as the block length increases.
The bound applies to any nonsingular memoryless channel such that
$(Z(\eta),Z'(\eta))$ is strongly nonlattice.
The main difference from other analyses is that
we optimize the parameter $\la$ around $\eta$ depending on
the sent and the received sequences $(\bmX,\bmY)$.
A future work is to extend the bound
to the case that
$(Z(\eta),Z'(\eta))$ is not strongly nonlattice, that is,
$(Z(\eta),Z'(\eta))$ is distributed on a set of lattice points
or on a set of parallel lines with an equal interval.
It may be possible to derive
an expression of asymptotic expansion applicable
to our problem by following the discussion in \cite[Chap.~5]{ranga}.

\section*{Acknowledgment}
The author thanks the anonymous reviewers
for their helpful comments and suggestion on many related works.
This work was supported in part by JSPS KAKENHI Grant Number 26106506.

\bibliographystyle{IEEEtran}
\bibliography{bunken}

\appendix

\subsection{Properties of Function $g_h$}\label{append_gh}
\begin{lemma}
For $\ch=1+h\eta$ it holds that
\begin{align}
g_h(u)
&\le
\min\{1,\ch u\}\label{upper_gh2}\\
&\le
\ch u^{\rho}\label{upper_gh}
\end{align}
and
\begin{align}
0\le\bibun{g_h(u)}{u}
&\le
(u+h\eta)\e^{-u}\label{g_bibun1}\\
&\le
\ch\label{g_bibun2}
\per
\end{align}
\end{lemma}
\begin{proof}
We obtain \eqref{upper_gh2} by
\begin{align}
g_h(u)
&=
1-\frac{\e^{-\frac{h\eta}{\e^{h\eta}-1}u}(1-\e^{-h\eta u})}{h\eta u}\nn
&\le
1-\frac{\e^{-u}\e^{-h\eta u}(\e^{h\eta u}-1)}{h\eta u}\nn
&\le
1-\e^{-(1+h\eta)u}\nn
&\le
\min\{1,\ch u\}\n
\end{align}
and \eqref{upper_gh} is straightforward from $0<\rho\le 1$.
We obtain \eqref{g_bibun1} by
\begin{align}
\bibun{g_h(u)}{u}
&=
\e^{-\frac{h\eta u}{e^{h\eta}-1}}
\left(
\frac{1-\e^{-h\eta u}}{\e^{h\eta}-1}+\frac{1-\e^{-h\eta u}(1+h\eta u)}{h\eta u^2}
\right)\nn
&\le
\e^{-u}
\left(
\frac{h\eta u}{h\eta}+\frac{1-(1-h\eta u)(1+h\eta u)}{h\eta u^2}
\right)\nn
&=
(u+h\eta)\e^{-u}\n
\end{align}
and \eqref{g_bibun2} follows from $u\e^{-u}\le 1$ for any $u\ge 0$.
\end{proof}

\subsection{Proof of Lemma \ref{thm_tool_lattice}}\label{append_tool_lattice}
The proof of Lemma \ref{thm_tool_lattice} is almost the same as
\cite[Thm.\,3.7.4]{LDP}
where the same result is proved for the i.i.d.~case based on
the asymptotic expansion for i.i.d.~random variables.


In \cite[Thm.\,2, Sect.\,XVI]{feller_vol2},
the asymptotic expansion for one-dimensional lattice random variables
is derived for i.i.d.~cases.
It is discussed in \cite[Sect.\,XVI.6.6]{feller_vol2} that
the result is easily extended to non-i.i.d.~cases
by slightly modifying the proof with some examples
depending on regularity conditions.
In our setting the following expression is
convenient as an asymptotic expansion
for non-i.i.d.~lattice random variables.
\begin{proposition}\label{prop_exp_lattice}
Let $\ep,\ua_2,\oa_2,\ua_3,\oa_3,\oa_4,\gapxia,\dephi>0$ be arbitrary
and $V_1,\cdots,\allowbreak V_n\in\mathbb{R}$ be
independent lattice
random variables
such that
the greatest common divisor of their spans is $h$,
$\E[V_i]=0$ and $\Pr[V_i/h\in \mathbb{Z}]=1$.
Then there exists
$\gapxi=\gapxi(\ua_2,\oa_2,\ua_3,\oa_3,\oa_4,\gapxia),
n_0=n_0(\ep,\ua_2,\oa_2,\ua_3,\oa_3,\oa_4,\gapxia,\dephi)$
satisfying the following:
it holds for all $n\ge n_0$ satisfying
\begin{align}
&n \ua_2\le\sum_{i=1}^n V_i^2
\le n\oa_2\com\nn
&n \ua_3\le\sum_{i=1}^n V_i^3
\le n\oa_3\com\nn
&\sum_{i=1}^n \log |\E[\e^{\im \xi V_i}]|\le -n\dephi,\quad
\forall \xi \in [-\pi/h,\pi/h]\setminus[-\gapxi,\gapxi]\com\nn
&\sum_{i=1}^n\left|\bibun{^4\log \E[\e^{\im \xi V_i}]}{\xi^4}\right|
\le n\oa_4,
\qquad \forall |\xi|\le \gapxia
\n
\end{align}
that
\begin{align}
&\sup_v\Bigg|\Pr\left[
\frac{\sum_{i=1}^nV_i}{\sqrt{ns_2}}\le v
\right]
-\Phi(v)
-\frac{s_3}{6\sqrt{n}}(1-v^2)\phi(v)
\nn
&\phantom{wwwwwwwwwwwwww}
-\phi(v)\tau\left(v,\frac{h}{\sqrt{nA_2}}\right)\Bigg|
\le \frac{\ep}{\sqrt{n}}\com\n
\end{align}
where
$s_m=n^{-1}\sum_{i=1}^n V_i^m$,
$\tau(v,d)=d\lceil v/d\rceil-v-d/2$,
$\Phi$ and $\phi$ are the cumulative distribution function and the density
of the standard normal distribution.
\end{proposition}
\begin{proof}[Proof of Lemma \ref{thm_tool_lattice}]
Let 
$P'$ be the probability distribution
of $\{V_i\}$ such that
$\rd P'/\rd P=\e^{\la^* \sum_{i=1}^nV_i}/\e^{\La_{\bm{V}}(\la^*)}$.
Then
\begin{align}
P\left[\sum_{i=1}^n V_i\ge x\right]
&=
\e^{-\La_{\bm{V}}(\la^*)}
\E_{P'}\left[\e^{\la^*\sum_{i=1}^n V_i}\idx{\sum_{i=1}^n V_i\ge x}\right]\per\n
\end{align}
Here note that
\begin{align}
\E_{P'}[V_i]=
\frac{\E[V_i\e^{\la^* V_i}]}{\e^{\La_{V_i}(\la^*)}}\n
\end{align}
and
\begin{align}
\sum_{i=1}^n\frac{\E[V_i\e^{\la^* V_i}]}{\e^{\la^*V_i}}
=\La_{\bm{V}}'(\la^*)=x\n
\end{align}
from the definition of $\la^*$.
Therefore
\begin{align}
\lefteqn{
P\left[\sum_{i=1}^n V_i\ge x\right]=
}\nn
&
\e^{-\La_{\bm{V}}(\la^*)}
\E_{P'}\left[\e^{\la^*\sum_{i=1}^n V_i}\idx{\sum_{i=1}^n (V_i-\E_{P'}[V_i])\ge 0}\right]\!.
\label{to_dembo}
\end{align}
Here the variance of $V_i$ under $P'$ are represented by
\begin{align}
\E_{P'}[(V_i-\E_{P'}[V_i])^2]
&=
\frac{\rd^2 \La_{V_i}(\la)}{\rd \la^2}\bigg|_{\la=\la^*}
\n
\end{align}
and similarly
\begin{align}
\sum_{i=1}^n \log |\E_{P'}[\e^{\im \xi V_i}]|
&=
\sum_{i=1}^n \log \left|\frac{\E[\e^{\la^* V_i}\e^{\im \xi V_i}]}{\e^{\La(\la^*)}}\right|\nn
&=
\sum_{i=1}^n \left(\log |\E[\e^{(\la^*+\im\xi) V_i}]|-\log \E[\e^{\la^* V_i}]\right)\per\n
\end{align}
Thus we can apply Prop.\,\ref{prop_exp_lattice}
to the evaluation of \eqref{to_dembo}
and
we obtain Lemma \ref{thm_tool_lattice}
by the same argument as \cite[Thm.\,3.7.4]{LDP} for the i.i.d.~case.
\end{proof}

\subsection{Proof of Lemma \ref{lem_regularity}}\label{append_regularity}
In this appendix we show Lemma \ref{lem_regularity}.
Note that \eqref{regularity_kantan1} is obtained
easily by the standard discussion used in the derivation of
random coding exponent and \eqref{regularity_kantan2} also easily follows from
\eqref{upper_gh}.

We prove Lemma \ref{lem_regularity} based on Cram\'er's theorem in \cite{LDP}
for general vector spaces, which is written for our setting
as follows\footnote{Cram\'er's theorem in \cite{LDP} is
described for a more general setting such that $\mathcal{V}$ is sufficient to
be a metric space under some regularity conditions.
When we consider Banach spaces some of these conditions are satisfied
and the theorem can be represented in the form of this paper.}.
\begin{proposition}[{Cram\'er's theorem \cite[Theorem 6.1.3]{LDP}}]
Let $\mu$ denote the distribution of i.i.d.\,random variables $V_1,V_2,\cdots$
on a topological real vector space $\mathcal{V}$.
Assume that $\mathcal{V}$ is a separable Banach space.
Then, for any compact set $\mathcal{S}\subset \mathcal{V}$,
\begin{align}
\lefteqn{
\blimsup_{n\to\infty}\frac1n\log \Pr\left[
\frac1n\sum_{i=1}^n V_i \in \mathcal{S}
\right]
}\nn
&\le -\inf_{v\in \mathcal{S}}\sup_{\theta \in \mathcal{V}^*}
\{\inner{v,\theta}-\log\E[\e^{\inner{V_1,\theta}}]\}\com
\n
\end{align}
where $\mathcal{V}^*$ is the topological dual of $\mathcal{V}$.
\end{proposition}
We use the following lemma
derived from this proposition.

\begin{lemma}\label{cramer_benri}
Let $\mathcal{V}$ be the space of continuous functions
on a compact set $\mathcal{A}$ into $\mathbb{R}$
and $V_1,\cdots,V_n$ be i.i.d.~random variables on $\mathcal{V}$
such that
$\E[V(s)]=v(s)$ and $\sup_{s\in \mathcal{S}}\E[\e^{\alpha_0 |V(s)|}]<\infty$
for some $\alpha_0>0$.
Then, for any compact set $\mathcal{A'}\subset \mathcal{A}$ and $\ep>0$, the empirical mean
$\bar{V}=n^{-1}\sum_{i=1}^nV_i$ satisfies
\begin{align}
\blimsup_{n\to\infty}\frac{1}{n}\log \Pr\left[
\sup_{s\in\mathcal{A}'}|\bar{V}(s)-v(s)|
\ge \ep\right]<0\per\n
\end{align}

\end{lemma}
\begin{proof}
Let $\mathcal{V}\ni f$ be equipped with the max norm
\begin{align}
\Vert f\Vert=\max_{s\in \mathcal{S}}|f(s)|\n
\end{align}
and $\mathcal{V}^*$ be its topological dual, that is,
the family of (signed) finite Borel measures on $\mathcal{S}$.
Then, we obtain from Cram\'er's theorem
for $\mathcal{S}=\{f\in\mathcal{V}:\sup_{s\in \mathcal{A}'}|f(s)-v(s)|\ge \ep\}$
that
\begin{align}
\lefteqn{
\blimsup_{n\to\infty}
\frac1n\log
\Pr\left[\sup_{s\in\mathcal{A}'}|\bar{V}(s)-v(s)|\ge \ep\right]
}\nn
&\le
-\inf_{f\in\mathcal{S}}\sup_{\theta\in \mathcal{V}^*}
\{\inner{f,\theta}-\log \E[\e^{\inner{V_1,\theta}}]\}\per\n
\end{align}

By considering a set of point mass measures
$\{\alpha \de_{\{s\}}:\alpha\in\mathbb{R},s\in\mathcal{A}\}$
as a subset of $\mathcal{V}^*$,
we obtain
\begin{align}
\lefteqn{
\inf_{f\in\mathcal{S}}\sup_{\theta\in \mathcal{V}^*}
\{\inner{f,\theta}-\log \E[\e^{\inner{V_1,\theta}}]\}\n
}\nn
&\ge
\inf_{f\in\mathcal{S}}
\sup_{s\in \mathcal{A}'}
\sup_{\alpha}
\left\{
\alpha f(s)-\log \E[\e^{\alpha V(s)}]
\right\}\per\n
\end{align}
Here note that
\begin{align}
0&<\frac{\partial{^2}}{\partial\alpha^2}\log \E\left[\e^{\alpha V(s)}\right]\nn
&\le \frac{\E[V(s)^2 \e^{\alpha V(s)}]}{\E[\e^{\alpha V(s)}]}\nn
&\le \frac{\E[V(s)^2 \e^{\alpha V(s)}]}{\E[1+\alpha V(s)]}\nn
&\le \frac{\E[V(s)^2 \e^{\alpha |V(s)|}]}{1-|\alpha| \E[|V(s)|]}\n
\end{align}
for $|\alpha|<1/\E[|V(s)|]$.
Since there exists $\beta>0$ such that
$x^2 \e^{\alpha_0 |x|/2}\le \beta (\e^{\alpha_0 |x|}+1)$
and $|x|\le \beta\e^{\alpha_0 |x|}$ hold for all $x\in\mathbb{R}$,
\begin{align}
\sup_{|\alpha|<\alpha_0/2}\hen{^2}{\alpha}\log \E\left[\e^{\alpha V(s)}\right]
<c\n
\end{align}
for some $c>0$.
Therefore
\begin{align}
\lefteqn{
\inf_{f\in\mathcal{S}}\sup_{\theta\in \mathcal{V}^*}
\{\inner{f,\theta}-\log \E[\e^{\inner{V_1,\theta}}]\}\n
}\nn
&\ge
\inf_{f\in\mathcal{S}}
\sup_{s\in \mathcal{A}'}
\sup_{|\alpha|\le \alpha_0/2}
\left\{
\alpha f(s)-\alpha V(s)-c\alpha^2/2
\right\}\nn
&\ge
\inf_{f\in\mathcal{S}}
\sup_{|\alpha|\le \alpha_0}
\left\{
|\alpha| \ep-c\alpha^2/2
\right\}\nn
&>0\n
\end{align}
and we obtain the lemma.
\end{proof}

We can apply Lemma \ref{cramer_benri}
to the proof of Lemma \ref{lem_regularity}
from the following lemma.

\begin{lemma}\label{lem_nonlattice_cumulant}
Let $\la>0$ and $\xi\in [-\pi/h,\pi/h]\setminus\{0\}$ be arbitrary.
If $\nu$ satisfy the lattice condition
then
\begin{align}
\E_{\rho}[\cumu(\la+\im\xi)]-\E_{\rho}[\cumu(\la)]<0\per\n
\end{align}
\end{lemma}
\begin{proof}
Let $\E_{X',\la}$ be the conditional expectation on $X'$ given $(X,Y)$
under distribution $P_{X',\la}$
such that $\rd P_{X',\la}/\rd P_{X'}=\e^{\la r(X,Y,X')}/\E_{X'}[\e^{\la r(X,Y,X')}]$.
Then
\begin{align}
\lefteqn{
\E_{\rho}[Z(\la+\im\xi)]-\E_{\rho}[Z(\la)]
}\nn
&=
\E_{\rho}\left[\log\frac{|\E_{X'}[\e^{(\la+\im\xi)r(X,Y,X')}]|}{\E_{X'}[\e^{\la r(X,Y,X')}]}\right]\nn
&=
\E_{\rho}\left[\log |\E_{X',\la}[\e^{\im\xi r(X,Y,X')}]|\right]\nn
&=
\E_{\rho}\left[\log |\E_{X',\la}[\e^{\im\xi\log \nu(X',Y)}]\e^{-\im\xi\log \nu(X,Y)}|\right]\nn
&=
\E_{\rho}\left[\log |\E_{X',\la}[\e^{\im\xi\log \nu(X',Y)}]|\right]
\per\n
\end{align}

On the other hand,
the definition of lattice condition in Def.\,\ref{def_lattice} implies that
$P[|\E_{X',\la}[\e^{\im\xi\log \nu(X',Y)}]|=1]<1$ holds
for any $\xi\notin\{2m\pi/h:m\in\mathbb{Z}\}$.

Since $P$ is absolutely continuous with respect to $P_{\rho}$
we have
$P_{\rho}[|\E_{X',\la}[\e^{\im\xi\log \nu(X',Y)}]|=1]<1$
for any $\xi\notin\{2m\pi/h:m\in\mathbb{Z}\}$.
Thus we obtain
$\E_{\rho}[\log |\E_{X',\la}[\e^{\im\xi\log \nu(X',Y)}]|]<0$
by noting that
$\E[V]<0$ holds for any random variable $V\in\mathbb{R}$ such that $V\le 0$ a.s.~and
$\Pr[V<0]>0$.
\end{proof}

\begin{proof}[Proof of Lemma \ref{lem_regularity}]
First we have
\begin{align}
\lefteqn{
\E_{XY}[\idx{S^c}
\e^{n\rho (\bar{Z}(\eta)+R)}
]}\nn
&=
\e^{n(\La(\rho)+\rho R)}
P_{\rho}[S^c]\nn
&\le
\e^{n(\La(\rho)+\rho R)}
\Bigg(
P_{\rho}[|\bar{Z}^{(1)}(\eta)|\ge \gap_1]
+
P_{\rho}[\bar{Z}^{(2)}(\la)\notin \mathcal{A}_2]\nn
&\qquad\quad+
P_{\rho}[\bar{Z}^{(3)}(\la)\notin \mathcal{A}_3]
+
P_{\rho}[\cumub(\la+\im\xi)-\cumub(\la)\in\mathcal{B}]\nn
&\qquad\quad+
P_{\rho}\left[\left|
\hen{^4}{\xi^4}
\bar{Z}^{(4)}(\la+\im\xi)
\right|\in\mathcal{C}\right]
\Bigg)\per\label{p_exp}
\end{align}
Note that the moment generating functions
of the absolute values of the empirical means in \eqref{p_exp} exist
from the regularity conditions assumed in \eqref{regularity_moment}.
It is straightforward from Cram\'er's inequality that
\begin{align}
\blimsup_{n\to\infty}\frac{1}{n}\log P_{\rho}[|\bar{Z}^{(1)}(\eta)|\ge \gap_1]<0\n
\end{align}
since $\E_{\rho}[Z^{(1)}(\eta)]=0$.
It is also straightforward from Lemmas \ref{cramer_benri} and \ref{lem_nonlattice_cumulant}
that the other four probabilities in \eqref{p_exp} are exponentially small
for sufficiently small $\gapla$ with respect to $(\gap_2,\gap_3)$
and
\begin{align}
\dephi&=-\frac{1}{2}
\sup_{\overset{\scriptstyle \la\in [\eta-\gapla,\eta+\gapla]}%
{\xi\in[-\pi/h,\pi/h]\setminus[-\gapxi,\gapxi]}}
\E_{\rho}[\cumu(\la+\im\xi)-\cumu(\la)]\nn
\oa_4&=2
\sup_{\overset{\scriptstyle \la\in [\eta-\gapla,\eta+\gapla]}%
{\xi\in[-\gapxia,\gapxia]}}
\E_{\rho}\left[
\left|\hen{^4Z(\la+\im\xi)}{\xi^4}\right|
\right]\per\n
\end{align}
\end{proof}

\subsection{Theorem \ref{thm_expansion1} for Nonlattice Channels}\label{append_nonlattice}
In this appendix
we give a brief explanation for the proof of Theorem \ref{thm_expansion1}
in the case that $h=0$, that is,
$\nu$ does not satisfy the lattice condition.
For this case we bound the error probability
by
\begin{align}
\E_{\bmX\bmY}[
\tilde{q}_{M}(\tilde{p}_{1/\sqrt{n}}(\bmX,\bmY))]
\le P_{\mathrm{RC}}
\le
\E_{\bmX\bmY}[
\tilde{q}_{M}(\tilde{p}_{0}(\bmX,\bmY))]\n
\end{align}
where
\begin{align}
\tilde{p}_{\zeta}(\bmx,\bmy)&=P_{\bmX'}[r(\bmx,\bmy,\bmX')\ge \zeta]\nn
\tilde{q}_M(p)&=
1-(1-p)^{M-1}\per\n
\end{align}
Similarly to Lemma \ref{thm_unify}
we have the following lemma.
\begin{lemma}\label{lem_error_nonlattice}
It holds for any $c\in(0,1/2)$ that
\begin{align}
%
\blimsup_{M\to\infty}
\sup_{p\in(0,1/2]}
\frac{\tilde{q}_M(p)}{1-\e^{-pM}}
=
\bliminf_{M\to\infty}
\inf_{p\in(0,1/2]}
\frac{\tilde{q}_M(p)}{1-\e^{-pM}}
=1\per\n
\end{align}
\end{lemma}
The proof of this lemma is given in Appendix \ref{proof_unify}.
We can obtain Theorem \ref{thm_expansion1} for $h=0$
by replacing the exact asymptotics for non-i.i.d.~lattice
random variables with
that for nonlattice random variables
based on the asymptotic expansion for
nonlattice random variables considered in \cite[Thm.\,1, Sect.\,XVI]{feller_vol2}.
More precisely we can show Theorem \ref{thm_expansion1} by
replacing Prop.\,\ref{prop_exp_lattice} with the following 
proposition, which is also easily obtain from the discussion
in \cite[Sect.\,XVI.6.6]{feller_vol2} for non-i.i.d.~random variables.

\begin{proposition}\label{prop_exp_nonlattice}
Let $\ep,\ua_2,\oa_2,\ua_3,\oa_3,\oa_4,\gapxia,\dephi>0$ be arbitrary
and $V_1,\cdots,\allowbreak V_n\in\mathbb{R}$ be
strongly nonlattice independent
random variables
such that
$\E[V_i]=0$ and $\Pr[V_i/h\in \mathbb{Z}]=1$.
Then there exists
$\underline{d}=\underline{d}(\ua_2,\oa_2,\ua_3,\oa_3,\oa_4,\gapxia)
<\overline{d}=\overline{d}(\ep,\ua_2,\oa_2,\ua_3,\oa_3,\oa_4,\gapxia)$
and
$n_0=n_0(\ep,\ua_2,\oa_2,\ua_3,\oa_3,\oa_4,\gapxia,\dephi)$
satisfying the following:
it holds for all $n\ge n_0$ satisfying
\begin{align}
&n \ua_2\le\sum_{i=1}^n V_i^2
\le n\oa_2\com\nn
&n \ua_3\le\sum_{i=1}^n V_i^3
\le n\oa_3\com\nn
&\sum_{i=1}^n \log |\E[\e^{\im \xi V_i}]|\le -n\dephi,\quad
\forall \xi \in [\underline{d},\overline{d}]\com\nn
&\sum_{i=1}^n\left|\bibun{^4\log \E[\e^{\im \xi V_i}]}{\xi^4}\right|
\le n\oa_4,
\qquad \forall |\xi|\le \gapxia
\n
\end{align}
that
\begin{align}
&\sup_v\Bigg|\Pr\left[
\frac{\sum_{i=1}^nV_i}{\sqrt{ns_2}}\le v
\right]
-\Phi(v)
\nn
&\phantom{wwwwwwwww}
-\frac{s_3}{6\sqrt{n}}(1-v^2)\phi(v)
\Bigg|
\le \frac{\ep}{\sqrt{n}}\per\n
\end{align}
\end{proposition}

\subsection{Bounds on Error Probability for $M$ Codewords}\label{proof_unify}
In this appendix we prove Lemmas \ref{thm_unify} and \ref{lem_error_nonlattice}.

\begin{proof}[Proof of Lemma \ref{thm_unify}]
First we have
\begin{align}
\lefteqn{
\sum_{i=1}^{M-1}\pz^i(1-\pz-\pp)^{M-i-1}{{M-1}\choose i}
}\nn
&=
(1-\pp)^{M-1}-
(1-\pz-\pp)^{M-1}\label{i1}
\end{align}
and
\begin{align}
\lefteqn{
\sum_{i=1}^{M-1}\pz^i(1-\pz-\pp)^{M-i-1}{{M-1}\choose i}\frac{1}{i+1}
}\nn
&=
\frac{1}{M}\sum_{i=1}^{M-1}\pz^i(1-\pz-\pp)^{M-i-1}{M \choose {i+1}}\nn
&=
\frac{1}{M\pz}\sum_{i=2}^{M}\pz^i(1-\pz-\pp)^{M-i}{M \choose {i}}\nn
&=
\frac{(1-\pp)^{M}-(1-\pz-\pp)^{M}}{M\pz}-(1-\pz-\pp)^{M-1}.\label{i2}
\end{align}
Combining \eqref{i1} and \eqref{i2} with \eqref{error_moto}
we obtain
\begin{align}
q_M(\pp,\pz)
&=
1
-\frac{(1-\pp)^{M}-(1-\pz-\pp)^{M}}{M\pz}\n
\end{align}
and
\begin{align}
\lefteqn{
1-\frac{
q_M(\pp,\pz)
}%
{1-\frac{\e^{-M\pp}(1-\e^{-M\pz})}{M\pz}}
}\nn
&=
1-\frac{M\pz-(1-\pp)^{M}-(1-\pz-\pp)^{M}}%
{M\pz-\e^{-M\pp}(1-\e^{-M\pz})}\nn
&=
1-\frac{M\pz-(1-\pp)^{M}\left(1-\left(1-\frac{\pz}{1-\pp}\right)^{M}\right)}%
{M\pz-\e^{-M\pp}(1-\e^{-M\pz})}\nn
&=
\frac{(1-\pp)^{M}\left(1-\left(1-\frac{\pz}{1-\pp}\right)^{M}\right)-\e^{-M\pp}(1-\e^{-M\pz})}%
{M\pz-\e^{-M\pp}(1-\e^{-M\pz})}.
\n
\end{align}
Here note that $\log (1-x)\ge -x-2x^2$ for $x\le 1/2$.
Therefore for $\pz,\pp\le 1/3$ we have
\begin{align}
\lefteqn{
(1-\pp)^{M}\left(1-\left(1-\frac{\pz}{1-\pp}\right)^{M}\right)
}\nn
&\le
\e^{-M\pp}\left(1-\e^{-\frac{M\pz}{1-\pp}-\frac{2M\pz^2}{(1-\pp)^2}}\right)\nn
&\le
\e^{-M\pp}\left(1-\e^{-M\pz-2M\pz\pp-5M\pz^2}\right)\nn
&\le
\e^{-M\pp}\left(1-(1-\min\{1,5M(\pp^2+\pp\pz)\})\e^{-M\pz}\right)\com\n
\end{align}
which implies
\begin{align}
\lefteqn{
\blimsup_{M\to\infty}
\sup_{(\pp,\pz)\in (0,1/3]^2 : \pp\le M^{c}\pz}
\left\{
1-\frac{
q_M(\pp,\pz)
}%
{1-\frac{\e^{-M\pp}(1-\e^{-M\pz})}{M\pz}}
\right\}
}\nn
&\le
\blimsup_{M\to\infty}
\sup_{\pz\in (0,1/3]}
\frac{\min\{1,10M^{1+2c}\pz^2\}}%
{M\pz-(1-\e^{-M\pz})}.
\phantom{wwwwwwwwwwww}
\nn
&=
\blimsup_{M\to\infty}
\sup_{\pz \in (0,1/3]}
\frac{1}{M^{1-2c}}\frac{\min\{1,10(M\pz)^2\}}%
{M\pz-(1-\e^{-M\pz})}\nn
&=0\per
\n
\end{align}

Similarly,
for $\pz,\pp\le 1/3$ we have
\begin{align}
\lefteqn{
(1-\pp)^{M}\left(1-\left(1-\frac{\pz}{1-\pp}\right)^{M}\right)
}\nn
&\ge
\e^{-M\pp-2M\pp^2}\left(1-\e^{-\frac{M\pz}{1-\pp}}\right)\nn
&\ge
\e^{-M\pp-2M\pp^2}\left(1-\e^{-M\pz}\right)\nn
&\ge
\e^{-M\pp}(1-\min\{1,2M\pp^2\})\left(1-\e^{-M\pz}\right)\n
\end{align}
and
\begin{align}
\lefteqn{
\bliminf_{M\to\infty}
\inf_{(\pp,\pz)\in (0,1/3]^2 : \pp\le M^{c}\pz}
\left\{
1-\frac{
q_M(\pp,\pz)
}%
{1-\frac{\e^{-M\pp}(1-\e^{-M\pz})}{M\pz}}
\right\}}\nn
&\ge
-
\blimsup_{M\to\infty}
\sup_{(\pp,\pz)\in (0,1/3]^2 : \pp\le M^{1+c}\pz}
\frac{
\min\{1,2M^{1+2c}\pz^2\}
}%
{M\pz-(1-\e^{-M\pz})}\nn
&=0\com
\n
\end{align}
which concludes the proof.
\end{proof}

\begin{proof}[Proof of Lemma \ref{lem_error_nonlattice}]
By letting $t(x)=x^{-1}\log(1-x)$
we have
\begin{align}
\frac{1-(1-p)^{M-1}}{1-\e^{-pM}}
&=
\frac{1-\e^{p(M-1)t(p)}}{1-\e^{-pM}}\nn
&=
1-
\frac{\e^{-pM}\left(\e^{p(M+(M-1)t(p))}-1\right)}{1-\e^{-pM}}\nn
&=
1-
\frac{\e^{p(M+(M-1)t(p))}-1}{\e^{pM}-1}\per\n
\end{align}
By $t(x)\le -1$,
the second term is bounded from above
as
\begin{align}
\frac{\e^{p(M+(M-1)t(p))}-1}{\e^{pM}-1}
&
\le
\frac{\e^{p}-1}{\e^{pM}-1}\nn
&
\le
\frac{\e^{p}-1}{pM}\nn
&\le \frac{\e-1}{M}\label{sita_nonlattice}
\end{align}
and bounded from below as
\begin{align}
\lefteqn{
\frac{\e^{p(M+(M-1)t(p))}-1}{\e^{pM}-1}
}\nn
&\ge
\frac{p(M+(M-1)t(p))}{\e^{pM}-1}\nn
&=
\frac{M(p+\log(1-p))}{\e^{pM}-1}
-
\frac{pt(p)}{\e^{pM}-1}\nn
&\ge
\frac{M(-2p^2)}{\e^{pM}-1}
\label{mainasu}\\
&\ge
-\frac{2}{M}\frac{(Mp)^2}{\e^{pM}-1}
\nn
&\ge
-\frac{2}{M}\com
\qquad \left(
\mbox{by }\frac{x^2}{\e^{x}-1}\le 1
\mbox{ for $x> 0$}\right)
\label{ue_nonlattice}
\end{align}
where we used $\log(1-p)\ge -p-2p^2$ for $p\in [0,1/2]$ and
$t(x)\le 0$ in \eqref{mainasu}.
We complete the proof by letting $M\to\infty$ in \eqref{sita_nonlattice} and \eqref{ue_nonlattice}.
\end{proof}

\subsection{Evaluation of Oscillations}\label{append_osci}
In this appendix we prove Lemma \ref{lem_osci}
on the oscillations of function $f_n$.
We first show Lemmas \ref{lem_osci_all} and \ref{lem_mathg} below.
\begin{lemma}\label{lem_osci_all}
For any set $S\subset \mathbb{R}^2$,
\begin{align}
\omega_{f_n}(S)&\le
\ch(c_2)^{-\rho}n^{-\rho/2}\sup_{z_2:z\in S}\e^{-\rho z_2^2/2c_1}\per\n
\end{align}
\end{lemma}
\begin{proof}
We can bound $f_n$ as
\begin{align}
f_n(z_1,z_2)
&=
\e^{-\sqrt{n}\rho z_1}g_h\left(
\frac{\e^{\sqrt{n}z_1-z_2^2/2c_1}}{c_2\sqrt{n}}
\right)\nn
&=
(c_2\sqrt{n})^{-\rho}\e^{-\rho z_2^2/2c_1}u^{-\rho}g_h\left(
u\right)\nn
&\quad\since{by letting $u=\frac{\e^{\sqrt{n}z_1-z_2^2/2c_1}}{c_2\sqrt{n}}$}\nn
&\le
\ch\left(c_2\sqrt{n}\right)^{-\rho}\e^{-\rho z_2^2/2c_1}\per
\since{by \eqref{upper_gh}}\n
\end{align}
Thus we obtain the lemma since $f_n(z)\ge 0$.
\end{proof}

\begin{lemma}\label{lem_mathg}
Let $u>0$ and $r \in [-1/2,1/2]$ be arbitrary.
Then
\begin{align}
&|g_h((1+r)u)-g_h(u)|\le \ch|r| u\com\label{ggap2}\\
&|g_h((1+r)u)-g_h(u)|\le \ch|r|\per\label{ggap1}
\end{align}
\end{lemma}
\begin{proof}
Eq.\,\eqref{ggap2} is straightforward from \eqref{g_bibun2}.
We obtain \eqref{ggap1} from
\begin{align}
\bibun{g_h((1+r)u)}{r}
&=
u\bibun{g_h(v)}{v}\bigg|_{v=(1+r)u}\nn
&\le
u((1+r)u+h\eta)\e^{-u}\since{by \eqref{g_bibun1}}\nn
&\le
6\e^{-2}+h\eta \e^{-1}\nn
&\le
\ch\per\n
\end{align}
\end{proof}

By using these lemmas we can evaluate
the oscillation of $f_n$ within a ball as follows.
\begin{lemma}\label{lem_osci_tight}
Assume $|z_2|\le c_1\sqrt{n}/2$.
Then,
for sufficiently large $n$,
\begin{align}
\omega_{f_n}(B_{\de_n}(z))
&\le
\frac{8 c_h}{c_2}\de_n\e^{(1-\rho)\sqrt{n}z_1}\com
\label{lem_tight2}\\
\omega_{f_n}(B_{\de_n}(z))
&\le
4(1+\ch)\sqrt{n}\de_n\e^{-\rho\sqrt{n}z_1}\per\label{lem_tight1}
\end{align}
\end{lemma}

\begin{proof}
First we obtain for $z'$ satisfying $\Vert z'-z\Vert \le \de_n$ and sufficiently large $n$ that
\begin{align}
|(z_2')^2-z_2^2|
&\le
|z_2'-z_2|(|z_2'|+|z_2|)\nn
&\le
|z_2'-z_2|(2|z_2|+|z_2-z_2'|)\nn
&\le
\de_n\left|c_1\sqrt{n}+\de_n\right|\nn
&\le
2c_1\de_n\sqrt{n}\per\since{by $\lim_{n\to\infty}\de_n=0$}\n
\end{align}
Let
$w=z_1-z_2^2/(2c_1\sqrt{n})$
and $w'=z_1'-(z_2')^2/(2c_1\sqrt{n})$.
Then
\begin{align}
\left|w'-w\right|
&\le
\left|z_1'-z_1\right|+\frac{|z_2^2-(z'_2)^2|}{2c_1\sqrt{n}}\nn
&\le
2\de_n\per\n
\end{align}
Therefore
we obtain for sufficiently large $n$
that
\begin{align}
\left|\frac{\e^{\rho\sqrt{n}w'}}{\e^{\rho\sqrt{n}w}}-1\right|
&\le
2(\rho\sqrt{n}|w'-w|)\le 2\sqrt{n}\de_n\nn
\left|\frac{\e^{\rho\sqrt{n}z_1'}}{\e^{\rho\sqrt{n}z_1}}-1\right|
&\le
2\sqrt{n}\de_n\n
\end{align}
since $\lim_{n\to\infty}\sqrt{n}\de_n=0$.
Therefore by letting $\de_n'=2\de_n\sqrt{n}$
and using \eqref{ggap2}
we obtain for sufficiently large $n$
that
\begin{align}
f_{n}(z')
&\le
(1+\de_n')\e^{-\rho\sqrt{n}z_1}
g_h\left(
\frac{(1+\de_n')\e^{\sqrt{n}w}}{c_2\sqrt{n}}
\right)\nn
&\le
(1+\de_n')\e^{-\rho\sqrt{n}z_1}
\left(
g_h\left(
\frac{\e^{\sqrt{n}w}}{c_2\sqrt{n}}
\right)
+
\frac{\ch\de_n'\e^{\sqrt{n}w}}{c_2\sqrt{n}}
\right)\com
\nn
f_{n}(z')
&\ge
(1-\de_n')\e^{-\rho\sqrt{n}z_1}
\left(
g_h\left(
\frac{\e^{\sqrt{n}w}}{c_2\sqrt{n}}
\right)
-
\frac{\ch\de_n'\e^{\sqrt{n}w}}{c_2\sqrt{n}}
\right).\n
\end{align}
We obtain \eqref{lem_tight2} from
these inequalities
by
\begin{align}
\!\!\!
\omega_{f_n}(B_{\de_n}(z))
&\le
2\de_n'\e^{-\rho\sqrt{n}z_1}
\left(
g_h\left(
\frac{\e^{\sqrt{n}w}}{c_2\sqrt{n}}
\right)
+
\frac{\ch\e^{\sqrt{n}w}}{c_2\sqrt{n}}
\right)\nn
&\le
4\de_n'\e^{-\rho\sqrt{n}z_1}
\frac{\ch\e^{\sqrt{n}w}}{c_2\sqrt{n}}\since{by \eqref{upper_gh2}}
\nn
&\le
4\de_n'\e^{(1-\rho)\sqrt{n}z_1}
\frac{\ch}{c_2\sqrt{n}}\per\n
\end{align}

Similarly
we obtain
from \eqref{ggap1}
that
\begin{align}
f_{n}(z')
&\le
(1+\de_n')\e^{-\rho\sqrt{n}z_1}
\left(g_h\left(
\frac{\e^{\sqrt{n}w}}{c_2\sqrt{n}}
\right)+\ch\de_n'
\right)\nn
f_{n}(z')
&\ge
(1-\de_n')\e^{-\rho\sqrt{n}z_1}
\left(g_h\left(
\frac{\e^{\sqrt{n}w}}{c_2\sqrt{n}}
\right)-\ch\de_n'
\right)\per\n
\end{align}
From these inequalities
we obtain \eqref{lem_tight1} by
\begin{align}
\omega_{f_n}(B_{\de_n}(z))
&\le
2\de_n'\e^{-\rho\sqrt{n}z_1}
\left(g_h\left(
\frac{\e^{\sqrt{n}w}}{c_2\sqrt{n}}
\right)
+\ch\right)\nn
&\le
2(1+\ch)\de_n'\e^{-\rho\sqrt{n}z_1}
\per\n
\end{align}
\end{proof}

\begin{proof}[Proof of Lemma \ref{lem_osci}]
Let $b_n$ be such that
\begin{align}
\e^{\sqrt{n}b_n}=n^{1/2+1/4\rho}\de_n^{1/2\rho}\per\n
\end{align}

First we have
\begin{align}
\lefteqn{
\int \omega_f(\{z':\Vert z'-z\Vert\le \de \})\phi_{\Sigma}(z+a)\rd z
}\nn
&\le
\int_{|z_2|\le c_1\sqrt{n}/2,z_1\le b_n} \omega_f(B_{\de_n}(z))\phi_{\Sigma}(z+a)\rd z\nn
&\quad+
\int_{|z_2|\le c_1\sqrt{n}/2,z_1\ge b_n} \omega_f(B_{\de_n}(z))\phi_{\Sigma}(z+a)\rd z\nn
&\quad+
\int_{|z_2|\ge c_1\sqrt{n}/2} \omega_f(B_{\de_n}(z))\phi_{\Sigma}(z+a)\rd z\nn
&\le
\int_{|z_2|\le c_1\sqrt{n}/2,z_1\le b_n}
\frac{4c_4c_h}{c_2}\de_n\e^{(1-\rho)\sqrt{n}z_1}
\phi_{\Sigma}(z+a)\rd z\nn
&\quad+
\int_{|z_2|\le c_1\sqrt{n}/2,z_1\ge b_n}
2c_4(1+\ch)\sqrt{n}\de_n\e^{-\rho\sqrt{n}z_1}
\phi_{\Sigma}(z+a)\rd z\nn
&\quad+
\int_{|z_2|\ge c_1\sqrt{n}/2}
\ch(c_2\sqrt{n})^{-\rho}\e^{-c_1\rho n/8}\phi_{\Sigma}(z+a)\rd z\nn
&\phantom{wwwwwwwwwwwwwwww}\since{by Lemmas \ref{lem_osci_all} and \ref{lem_osci_tight}}\nn
&\le
\int_{-\infty}^{b_n}
\frac{4c_4c_h}{c_2}\de_n\e^{(1-\rho)\sqrt{n}z_1}
\phi_{\si_{11}}(z_1+a_1)\rd z_1\nn
&\quad+
\int_{b_n}^{\infty}
2c_4(1+\ch)\sqrt{n}\de_n\e^{-\rho\sqrt{n}z_1}\phi_{\si_{11}}(z_1+a_1)\rd z_1
\nn
&\quad+
\so(n^{-(1+\rho)/2})\per\label{bunkatu_osci}
\end{align}
Here recall that $\lim_{n\to\infty}\sqrt{n}\de_n=0$
and therefore the second term of \eqref{bunkatu_osci} is bounded as
\begin{align}
\lefteqn{
\int_{b_n}^{\infty} \sqrt{n}\de_n\e^{-\rho\sqrt{n}z_1}\phi_{\si_{11}}(z_1+a_1)\rd z_1
}\nn
&\le
\frac{1}{\sqrt{2\pi\si_1^2}}\int_{b_n}^{\infty}\sqrt{n}\de_n\e^{-\rho\sqrt{n}z_1}\rd z_1\nn
&=
\frac{1}{\sqrt{2\pi\si_1^2}}\frac{\de_n\e^{-\rho\sqrt{n}b_n}}{\rho}\nn
&=
\frac{1}{\sqrt{2\pi\si_1^2}}\frac{\de_n (n^{1/2+1/4\rho}\de_n^{1/2\rho})^{-\rho}}{\rho}\nn
&=
\frac{1}{\sqrt{2\pi\si_1^2}}\frac{(\sqrt{n}\de_n)^{1/2}}{\rho n^{(1+\rho)/2}}\nn
&=
\so(n^{-(1+\rho)/2})\per\n
\end{align}
We obtain \eqref{osci2}
since the first term of \eqref{bunkatu_osci}
is bounded as
\begin{align}
\lefteqn{
\int_{-\infty}^{b_n}
\de_n\e^{(1-\rho)\sqrt{n}z_1}
\phi_{\si_{11}}(z_1+a_1)\rd z_1
}\nn
&\le
\de_n\e^{(1-\rho)\sqrt{n}b_n}\nn
&=
\de_n(\sqrt{n}(\sqrt{n}\de_n)^{1/2\rho})^{(1-\rho)}\nn
&=
n^{-\rho}(\sqrt{n}\de_n)(\sqrt{n}\de_n)^{(1-\rho)/2\rho}\nn
&=
\so(n^{-\rho})\per\n
\end{align}
We obtain \eqref{osci3}
since the first term of \eqref{bunkatu_osci} is also bounded
for $\rho<1$ as
\begin{align}
\lefteqn{
\int_{-\infty}^{b_n}
\de_n\e^{(1-\rho)\sqrt{n}z_1}
\phi_{\si_{11}}(z_1+a_1)\rd z_1
}\nn
&\le
\frac{1}{\sqrt{2\pi\si_{11}^2}}
\int_{-\infty}^{b_n}
\de_n\e^{(1-\rho)\sqrt{n}z_1}
\rd z_1\nn
&=
\frac{1}{\sqrt{2\pi\si_{11}^2}}
\frac{\de_n(\sqrt{n}(\sqrt{n}\de_n)^{1/2\rho})^{1-\rho}}{(1-\rho)\sqrt{n}}
\nn
&=
\frac{1}{\sqrt{2\pi\si_{11}^2}}
\frac{n^{-(1+\rho)/2}(\sqrt{n}\de_n)(\sqrt{n}\de_n)^{(1-\rho)/2\rho}}{1-\rho}
\nn
&=
\so(n^{-(1+\rho)/2})\per\n
\end{align}
\end{proof}

\subsection{Proof of Lemma \ref{lem_main}}\label{append_second}

\begin{figure*}[t]
\hrulefill

(i) $\rho<1,\,\De=0$.
\begin{align}
\lefteqn{
\iint
\left(1-\frac{h(z_1,z_2)}{\sqrt{n}}\right)\frac{\e^{-(z_1,z_2)\Sigma_{01}^{-1}(z_1,z_2)^T/2}}%
{2\pi\sqrt{|\Sigma|}}
\e^{-\sqrt{n}\rho z_1}
g_h\left(
\frac{
 \e^{\sqrt{n}z_1-z_2^2/2c_1}
}{c_2\sqrt{n}}
\right)\rd z_1 \rd z_2
}\nn
&=
\frac{(c_2\sqrt{n})^{-\rho}}{\sqrt{n}}
\iint
\left(1-\frac{h((w+z_2^2/2c_1+d_n)/\sqrt{n},z_2)}{\sqrt{n}}\right)\cdot\nn
&\qquad\qquad\qquad\qquad
\frac{\e^{-((w+z_2^2/2c_1+d_n)/\sqrt{n},z_2)\Sigma_{01}^{-1}((w+z_2^2/2c_1+d_n)/\sqrt{n},z_2)^T/2}}%
{2\pi\sqrt{|\Sigma_{01}|}}
\e^{-\rho w-\rho z_2^2/2c_1}
g_h\left(\e^{w}\right)\rd w \rd z_2
\nn
&\phantom{wwwwwwwwwwwwwwwwwwwwwwwwwwwwww}
\left(\mbox{by letting } \e^{w}=\frac{
 \e^{\sqrt{n}z_1-z_2^2/2c_1}
}{c_2\sqrt{n}}
\mbox{ and } d_n=\log c_2\sqrt{n}\right)\nn
&=
\frac{(c_2\sqrt{n})^{-\rho}}{\sqrt{n}}
\iint
\left(1+\so(1)\right)
\frac{\e^{-(0,z_2)\Sigma_{01}^{-1}(0,z_2)^T/2}}%
{2\pi\sqrt{|\Sigma_{01}|}}
\e^{-\rho w-\rho z_2^2/2c_1}
g_h\left(\e^{w}\right)\rd w \rd z_2
\nn
&\quad+
n^{-(1+\rho)/2}
\iint_{\max\{|w|,|z_2|\}\ge n^{1/5}}
\e^{-\rho w-\rho z_2^2/2c_1}
g_h\left(\e^{w}\right)\rd w \rd z_2\nn
&\quad
\cdot\lo\left(\sup_{w,z'}\left\{\left(1-\frac{h((w+(z_2')^2/2c_1+d_n)/\sqrt{n},z_2')}{\sqrt{n}}\right)
\frac{\e^{-((w+(z_2')^2/2c_1+d_n)/\sqrt{n},z_2')\Sigma_{01}^{-1}((w+(z_2')^2/2c_1+d_n)/\sqrt{n},z_2')^T/2}}
{2\pi\sqrt{|\Sigma_{01}|}}\right\}\right)\!\!\!\label{tyuu_bound}\\
&=
\frac{(c_2\sqrt{n})^{-\rho}(1+\so(1))}{\sqrt{n}}
\iint
\frac{\e^{-(0,z_2)\Sigma_{01}^{-1}(0,z_2)^T/2}}
{2\pi\sqrt{|\Sigma_{01}|}}
\e^{-\rho w-\rho z_2^2/2c_1}
g_h\left(\e^{w}\right)\rd w \rd z_2\nn
&\quad+
n^{-(1+\rho)/2}
\iint_{\max\{|w|,|z_2|\}\ge n^{1/5}}
\e^{-\rho w-\rho z_2^2/2c_1}
g_h\left(\e^{w}\right)\rd w \rd z_2
\cdot \lo(1)
\nn
&=
\frac{(c_2\sqrt{n})^{-\rho}}{2\pi\sqrt{n|\Sigma_{01}|}}
\int \e^{-z_2^2(\sigma_{00}/|\Sigma_{01}|+\rho/c_1)/2}\rd z_2
\int
\e^{-\rho w}
g_h\left(\e^{w}\right)\rd w
+
\lo\left(n^{-(1+\rho)/2}
\iint_{|z_2| \ge n^{1/5}}
\e^{-\rho w-\rho z_2^2/2c_1}
g_h\left(\e^{w}\right)\rd w \rd z_2\right)\nn
&\qquad+
\lo\left(n^{-(1+\rho)/2}
\iint_{|w|\ge  n^{1/5}}
\e^{-\rho w-\rho z_2^2/2c_1}
g_h\left(\e^{w}\right)\rd w \rd z_2\right)\nn
&=
\frac{(c_2\sqrt{n})^{-\rho}}{\sqrt{2\pi n(\si_{00}+\rho|\Sigma_{01}|/c_1)}}
\int
\e^{-\rho w}g_h\left(\e^{w}\right)\rd w+\so(n^{-\frac{1+\rho}{2}})
\com\label{fig1}
\end{align}
where \eqref{tyuu_bound} follows from
\begin{align}
&\!\!\!\!\!\!\!\!\!\!
\left(1-\frac{h((w+(z_2)^2/2c_1+d_n)/\sqrt{n},z_2)}{\sqrt{n}}\right)
\frac{\e^{-((w+(z_2)^2/2c_1+d_n)/\sqrt{n},z_2)\Sigma_{01}^{-1}((w+(z_2)^2/2c_1+d_n)/\sqrt{n},z_2)^T/2}}
{2\pi\sqrt{|\Sigma_{01}|}}
\nn
&\qquad=
(1+\so(1))
\frac{\e^{-(0,z_2)\Sigma_{01}^{-1}(0,z_2)^T/2}}
{2\pi\sqrt{|\Sigma_{01}|}}\n
\end{align}
for $(w,z_2)$ such that $\max\{|w|,|z_2|\}\le n^{1/5}$.

\hrulefill
\end{figure*}
\begin{figure*}
\hrulefill

(ii) $\rho=1,\,\De=0$.
\begin{align}
\lefteqn{
\iint
\left(1-\frac{h(z_1,z_2)}{\sqrt{n}}\right)\frac{\e^{-(z_1,z_2)\Sigma_{01}^{-1}(z_1,z_2)^T/2}}%
{2\pi\sqrt{|\Sigma_{01}|}}
\e^{-\sqrt{n}\rho z_1}
g_h\left(
\frac{
 \e^{\sqrt{n}z_1-z^2/2c_1}
}{c_2\sqrt{n}}
\right)\rd z_1 \rd z_2
}\nn
&=
(c_2\sqrt{n})^{-1}
\iint \left(1-\frac{h(w+(z_2^2/2c_1+d_n)/\sqrt{n},z_2)}{\sqrt{n}}\right)\cdot\nn
&\qquad\qquad\qquad
\frac{\e^{-(w+(z_2^2/2c_1+d_n)/\sqrt{n},z_2)\Sigma_{01}^{-1}(w+(z_2^2/2c_1+d_n)/\sqrt{n},z_2)^T/2}}%
{2\pi\sqrt{|\Sigma_{01}|}}
\e^{-z_2^2/2c_1}
\e^{-\sqrt{n} w}
g_h\left(\e^{\sqrt{n} w}\right)\rd w \rd z_2\nn
&\phantom{wwwwwwwwwwwwwwwwwwwwwwwwwwwwwwwwwwww}\left(\mbox{by letting } \e^{\sqrt{n}w}=\frac{
 \e^{\sqrt{n}z_1-z^2/2c_1}
}{c_2\sqrt{n}}\right)\nn
&=
(c_2\sqrt{n})^{-1}
\frac{h\eta}{\e^{h\eta}-1}
\iint_{w\le -n^{-1/4}}\left(1+\so(1)\right)
\frac{\e^{-(w,z_2)\Sigma_{01}^{-1}(w,z_2)^T/2}}
{2\pi\sqrt{|\Sigma_{01}|}}
\e^{-z_2^2/2c_1}
\rd w \rd z_2
+\so(n^{-1/2})
\nn
&=
(c_2\sqrt{n})^{-1}
\frac{h\eta}{\e^{h\eta}-1}
\frac{1}{2\sqrt{|\Sigma_{01}|
\left|\Sigma_{01}^{-1}+\left(
\begin{array}{cc}
0&0\\0&1/c_1
\end{array}
\right)\right|}
}
+\so(n^{-1/2})\nn
&=
(c_2\sqrt{n})^{-1}
\frac{h\eta}{\e^{h\eta}-1}
\frac{1}{2\sqrt{1+\si_{11}/c_1}
}
+\so(n^{-1/2})\per\label{fig2}
\end{align}
(iii) $\rho=1,\,\De>0$.
\begin{align}
\lefteqn{
\iint
\frac{\e^{-(z_1,z_2)\Sigma_{01}^{-1}(z_1,z_2)^T/2}}{2\pi\sqrt{|\Sigma_{01}|}}
\left(1-\frac{h(z_1,z_2)}{\sqrt{n}}\right)
\e^{-\sqrt{n}\rho(z_1-\sqrt{n}\De)}
g_h\left(
\frac{
 \e^{\sqrt{n}(z_1-\sqrt{n}\De)-z^2/2c_1)}
}{c_1\sqrt{n}}
\right)\rd z_1 \rd z_2
}\nn
&=
(c_2\sqrt{n})^{-1}
\iint \left(1-\frac{h(w+(z_2^2/2c_1+d_n)/\sqrt{n},z_2)}{\sqrt{n}}\right)\cdot\nn
&\qquad\qquad\qquad
\frac{\e^{-(w+(z_2^2/2c_1+d_n)/\sqrt{n},z_2)\Sigma_{01}^{-1}(w+(z_2^2/2c_1+d_n)/\sqrt{n},z_2)^T/2}}%
{2\pi\sqrt{|\Sigma_{01}|}}
\e^{-z_2^2/2c_1}
\e^{-\sqrt{n}(w-\sqrt{n}\De)}
g_h\left(\e^{\sqrt{n} (w-\sqrt{n}\De)}\right)\rd w \rd z_2\nn
&\phantom{wwwwwwwwwwwwwwwwwwwwwwwwwwwwwwwwwwww}\left(\mbox{by letting } \e^{\sqrt{n}w}=\frac{
 \e^{\sqrt{n}z_1-z^2/2c_1}
}{c_2\sqrt{n}}\right)\nn
&=
(c_2\sqrt{n})^{-1}
\frac{h\eta}{\e^{h\eta}-1}
\iint_{w\le \sqrt{n}\Delta-n^{-1/4}}\left(1+\so(1)\right)
\frac{\e^{-(w,z_2)\Sigma_{01}^{-1}(w,z_2)^T/2}}
{2\pi\sqrt{|\Sigma_{01}|}}
\e^{-z_2^2/2c_1}
\rd w \rd z_2
+\so(n^{-1/2})\nn
&=
(c_2\sqrt{n})^{-1}
\frac{h\eta}{\e^{h\eta}-1}
\frac{1}{\sqrt{1+\si_{11}/c_1}
}
+\so(n^{-1/2})
\per\label{fig3}
\end{align}
\hrulefill
\end{figure*}

First we have
\begin{align}
\lefteqn{
\E\!\left[
g_h\left(
\frac{
 \e^{n(\bar{Z}(\eta)+R-(\bar{Z}'(\eta))^2/2c_1)}
}{c_2\sqrt{n}}
\right)
\right]
}\nn
&=
\e^{n\La(\rho)}
\E_{\rho}\left[
\e^{-n\rho\bar{Z}(\eta)}
g_h\left(
\frac{
 \e^{n(\bar{Z}(\eta)+R-(\bar{Z}'(\eta))^2/2c_1)}
}{c_2\sqrt{n}}
\right)
\right].\n
\end{align}
Here recall that
$\E_{\rho}[\bar{Z}(\eta)]=\mu_0\le -R$ and $\E_{\rho}[\bar{Z}'(\eta)]=\mu_1=0$
from \eqref{la_bibun}.
By letting $\De=-(R+\mu_0)$,
we have
$\De=0$ for $R\ge R_{\mathrm{crit}}$ and $\De>0$ for $R<R_{\mathrm{crit}}$.
Normalizing $\bar{Z}(\eta)$ and $\bar{Z}'(\eta)$
as
$\tilde{Z}_1=\sqrt{n}(\bar{Z}(\eta)+R+\De)$ and $\tilde{Z}_2=\sqrt{n}\bar{Z}'(\eta)$,
respectively, we have
\begin{align}
\lefteqn{
\E\!\left[
g_h\left(
\frac{
 \e^{n(\bar{Z}(\eta)+R-(\bar{Z}'(\eta))^2/2c_1)}
}{c_2\sqrt{n}}
\right)
\right]
}\nn
&=
\e^{-nE_r(R)}
\nn&\quad\quad\cdot
\E_{\rho}\left[
\e^{-\sqrt{n}\rho(\tilde{Z}_1-\sqrt{n}\De)}
g_h\left(
\frac{
 \e^{\sqrt{n}(\tilde{Z}_1-\sqrt{n}\De)-\tilde{Z}^2/2c_1)}
}{c_2\sqrt{n}}
\right)
\right]\!.\n
\end{align}
We obtain from Prop.\,\ref{prop_ranga}
that
\begin{align}
\lefteqn{
\E_{\rho}\left[
\e^{-\sqrt{n}\rho(\tilde{Z}_1-\sqrt{n}\De)}
g\left(
\frac{
 \e^{\sqrt{n}(\tilde{Z}_1-\sqrt{n}\De)-\tilde{Z}^2/2c_1)}
}{c_2\sqrt{n}}
\right)
\right]
}\nn
&=
\iint
\frac{\e^{-z^T\Sigma_{01}^{-1}z/2}}{2\pi\sqrt{|\Sigma|}}
\left(1-\frac{h(z)}{\sqrt{n}}\right)
\e^{-\sqrt{n}\rho(z_1-\sqrt{n}\De)}\nn
&\qquad
\cdot g\left(
\frac{
 \e^{\sqrt{n}(z_1-\sqrt{n}\De)-z_2^2/2c_1)}
}{c_1\sqrt{n}}
\right)\rd z_1 \rd z_2+
\omega_{f_n}(\de_n;\Phi)
\per\n
\end{align}

For the case (i) $\rho<1,\De=0$,
this integral is evaluated as \eqref{fig1}.
Similarly for cases
(ii) $\rho=1,\De=0$ and (iii) $\rho=1,\,\De>0$,
it is evaluated as \eqref{fig2} and \eqref{fig3},
respectively,
since
$\e^{-\sqrt{n}w}g_h(\e^{\sqrt{n}w})\le \e^{-\sqrt{n}w}$
holds for any $w$
and
\begin{align}
\e^{-\sqrt{n}w}g_h(\e^{\sqrt{n}w})=\frac{h\eta(1+\so(1))}{\e^{h\eta}-1}\n
\end{align}
holds for $w\le -n^{-1/4}$.


\noindent
\hrulefill

\noindent
(See the next two pages for
Eqs.\,\eqref{fig1}--\eqref{fig3}.
)\vspace{-1.5mm}

\noindent
\hrulefill

\noindent
Now, combined with
Lemma \ref{lem_osci},
it suffices to show that
\begin{align}
\int_{-\infty}^{\infty} \e^{-\rho w}g_h(\e^w)\rd w
&=\int_{0}^{\infty} z^{-(1+\rho)}g_h(z)\rd z\nn
&=\frac{1}{\rho}\int_{0}^{\infty} z^{-\rho}\bibun{g_h(z)}{z}\rd z\nn
&=\psi_{\rho,h}
\per\label{int_suffice}
\end{align}
By letting $a=h\eta$ and $b=a/(\e^a-1)$,
we can evaluate this integral as
\begin{align}
\lefteqn{
\int_{0}^{\infty} z^{-\rho}\bibun{g_h(z)}{z}\rd z
}\nn
&=
\int_0^{\infty} z^{-\rho-1}\frac{b\e^{-bz}-(a+b)\e^{-(a+b)z}}{a}\rd z\nn
&\quad+\int_0^{\infty} z^{-\rho-2}\frac{\e^{-bz}-\e^{-(a+b)z}}{a}\rd z
\label{int_bunkatu}
\end{align}
Here the first term is evaluated by integration by parts as
\begin{align}
\lefteqn{
\int_0^{\infty} z^{-\rho-1}\frac{b\e^{-bz}-(a+b)\e^{-(a+b)z}}{a}\rd z
}\nn
&=
\frac{1}{\rho}\int_0^{\infty} z^{-\rho}\frac{(a+b)^2\e^{-(a+b)z}-b^2\e^{-bz}}{a}\rd z\nn
&=
\frac{\Gamma(1-\rho)}{\rho}\frac{(a+b)^{\rho+1}-b^{\rho+1}}{a}\com
\label{integral1}
\end{align}
where we used the fact that for any $c>0$
\begin{align}
\int_0^{\infty} \e^{-cz}z^{-\rho}\rd z
=\Gamma(1-\rho)c^{\rho-1}\per\n
\end{align}

Similarly we have
\begin{align}
\lefteqn{
\int_0^{\infty} z^{-\rho-2}\frac{\e^{-bz}-\e^{-(a+b)z}}{a}\rd z
}\nn
&=\frac{1}{\rho+1}
\int_0^{\infty} z^{-\rho-1}\frac{-b\e^{-bz}+(a+b)\e^{-(a+b)z}}{a}\rd z\nn
&=\frac{1}{\rho(\rho+1)}
\int_0^{\infty} z^{-\rho}\frac{b^2\e^{-bz}-(a+b)^2\e^{-(a+b)z}}{a}\rd z\nn
&=\frac{\Gamma(1-\rho)}{\rho(\rho+1)}
\frac{b^{\rho-1}-(a+b)^{\rho-1}}{a}\per
\label{integral2}
\end{align}

Combining \eqref{int_bunkatu} with \eqref{integral1} and \eqref{integral2}
we obtain \eqref{int_suffice} by
\begin{align}
\lefteqn{
\int_{0}^{\infty} z^{-\rho}\bibun{g_h(z)}{z}\rd z
}\nn
&=
\frac{\Gamma(1-\rho)}{\rho}\frac{(a+b)^{\rho+1}-b^{\rho+1}}{a}
\left(1-\frac{1}{1+\rho}\right)\nn
&=
\frac{\Gamma(1-\rho)}{1+\rho}\frac{\left(\frac{h\eta\e^{h\eta}}{\e^{h\eta}-1}\right)^{\rho+1}-\left(\frac{h\eta}{\e^{h\eta}-1}\right)^{\rho+1}}{h\eta}\nn
&=
\frac{\Gamma(1-\rho)}{h\eta(1+\rho)}
\left(\frac{h\eta}{\e^{h\eta}-1}\right)^{\rho+1}
\left(\e^{h\eta(1+\rho)}-1\right)\nn
&=
\Gamma(1-\rho)
\left(\frac{h\eta}{\e^{h\eta}-1}\right)^{\rho+1}
\frac{\e^{h}-1}{h}
=\rho \psi_{\rho,h}
\com\n
\end{align}
where we used $\eta=1/(1+\rho)$.
\qed

%


\end{document}